\documentclass[sigconf]{acmart}

\AtBeginDocument{%
  }

\setcopyright{acmlicensed}
\copyrightyear{2018}
\acmYear{2018}
\acmDOI{XXXXXXX.XXXXXXX}
\acmConference[Conference acronym 'XX]{Make sure to enter the correct
  conference title from your rights confirmation email}{June 03--05,
  2018}{Woodstock, NY}
\acmISBN{978-1-4503-XXXX-X/2018/06}

\usepackage{enumitem}
\usepackage{multirow}
\usepackage{amsmath}
\usepackage{amsthm}
\usepackage{subfigure}
\usepackage{arydshln}
\usepackage{tabularray}
\usepackage{xcolor, colortbl}
\usepackage{graphicx}
\usepackage{subcaption}
\usepackage{caption}

\newcommand\blfootnote[1]{%
  \begingroup
  \renewcommand\thefootnote{}\footnote{#1}%
  \addtocounter{footnote}{-1}%
  \endgroup
}

\begin{document}

\settopmatter{printacmref=false, printfolios=false, printccs=false}
\renewcommand\footnotetextcopyrightpermission[1]{} 




\title[Generative Recommendation for Large-Scale Advertising]{Generative Recommendation for Large-Scale Advertising}

\author{Ben Xue*, Dan Liu*, Lixiang Wang*, Mingjie Sun*, Peng Wang*, Pengfei Zhang*, Shaoyun Shi*, Tianyu Xu*, Yunhao Sha*, Zhiqiang Liu*,    Bo Kong, Bo Wang, Hang Yang, Jieting Xue, Junhao Wang, Shengyu Wang, Shuping Hui, Wencai Ye, Xiao Lin, Yongzhi Li, Yuhang Chen, Zhihui Yin,\\    Quan Chen, Shiyang Wen, Wenjin Wu, Han Li, Guorui Zhou, Changcheng Li, Peng Jiang$\dag$, Kun Gai}

\affiliation{%
  \institution{Kuaishou Technology}
  \city{Beijing}
  \country{China}
  \postcode{100080}
}

\renewcommand{\shortauthors}{Peng Jiang et al.}

\begin{abstract}

Generative recommendation has recently attracted widespread attention in industry due to its potential for scaling and stronger model capacity.
However, deploying real-time generative recommendation in large-scale advertising requires designs beyond large-language-model (LLM)-style training and serving recipes.
We present a production-oriented generative recommender co-designed across architecture, learning, and serving, named \textbf{GR4AD} (\textbf{G}enerative \textbf{R}ecommendation for \textbf{AD}dvertising). 
As for tokenization, GR4AD proposes UA-SID (Unified Advertisement Semantic ID) to capture complicated business information.
 Furthermore, GR4AD introduces LazyAR, a lazy autoregressive decoder that relaxes layer-wise dependencies for short, multi-candidate generation, preserving effectiveness while reducing inference cost, which facilitates scaling under fixed serving budgets. 
To align optimization with business value, GR4AD employs VSL (Value-Aware Supervised Learning) and proposes RSPO (Ranking-Guided Softmax Preference Optimization), a \emph{ranking-aware, list-wise} reinforcement learning algorithm that optimizes value-based rewards under list-level metrics for continual online updates. For online inference, we further propose dynamic beam serving, which adapts beam width across generation levels and online load to control compute. 
Large-scale online A/B tests show up to 4.2\% ad revenue improvement over an existing DLRM-based stack, with consistent gains from both model scaling and inference-time scaling. GR4AD has been fully deployed in Kuaishou advertising system with over 400 million users and achieves high-throughput real-time serving.

\end{abstract}






\maketitle

\blfootnote{* Authors contributed equally and are listed in alphabetical order.}
\blfootnote{$\dag$ Corresponding author.}

\vspace{-3em}
\section{Introduction}

Recent industrial efforts on recommender systems have started converting deep learning recommendation models (DLRMs) into generative recommenders, driven by their potential for scaling and stronger model capacity, with significant gains across a range of business scenarios~\cite{zhou2025onerec,zhou2025onerecv2,zhang2025gpr,hstu2024,MTGR2025,scaling2024seqrec}.

Despite encouraging progress, deploying real-time generative recommendation models in large-scale advertising systems remains challenging. The unique characteristics of advertising recommendation make a direct reuse of LLM techniques insufficient in several aspects.
(1) \textbf{Advertisement Tokenization}. 
Tokenization is fundamental to LLM generation and particularly challenging for advertising, where short-video creatives fuse video attributes, product details, and B2B advertiser metadata. Prior work incorporates multimodal cues or pretrained MLLMs for content understanding, but no end-to-end, fine-tuned advertisement LLM embedding exists. 
Furthermore, platforms expose salient business signals (e.g., conversion type, ads account) that are not readily captured by semantic content,
so jointly modeling multimodal, multi-granularity features for user interest and business value remains a core challenge.
(2) \textbf{Learning Paradigm}. Advertising recommendation optimizes \emph{ranked lists} under business objectives (e.g., eCPM) and list-wise metrics (e.g., NDCG), which are not well captured by per-item supervision. Existing approaches, largely following LLM-style training recipes, lack a ranking-aware, list-wise learning design tailored to online advertising learning~\cite{chen2025onesearch,zheng2024adapting,wang2024learnable,zhou2025openonerec}.
(3) \textbf{Real-Time Serving}. Generative recommendation does not remove the stringent serving constraints of DLRM-based stacks: the system must produce multiple high-quality candidates under high traffic and strict latency budgets. This setting differs fundamentally from interactive LLM usage, where decoding a single response can tolerate substantially longer latency. Beyond optimizations borrowed from LLM inference, serving-time efficiency for real-time multi-candidate generation in advertising remains under-explored in a systematic way.

To address the above gaps, we present a production-oriented generative recommender designed for real-time, large-scale advertising, named \textbf{GR4AD} (\textbf{G}enerative \textbf{R}ecommendation for \textbf{AD}dvertising). GR4AD adopts a recommendation-native co-design across representation, learning, and serving.
(1) \textbf{Unified Advertisement Semantic ID}. We propose UA-SID, derived from a fine‑tuned MLLM embedding trained on real‑world advertisement creatives via instruction tuning and co‑occurrence learning to capture complicated information. Furthermore, we introduce MGMR (Multi‑Granularity-Multi‑Resolution) RQ‑Kmeans quantization method to model non semantic information, substantially reduce SID collisions, and improve codebook utilization.
(2) \textbf{Value-Aware Online Learning}. To align optimization with business value in non-stationary markets, we design VSL (Value-Aware Supervised Learning) to learn the user-interest distribution, and propose RSPO (Ranking-Guided Softmax Preference Optimization), a ranking-guided, list-wise RL algorithm that explicitly optimizes list-level objectives. We further develop an online learning framework that tightly integrates VSL and RL for continual, frequent model updates.
(3) \textbf{Recommendation-Oriented Efficiency Optimizations}. To satisfy strict real-time constraints, we systematically optimize decoding and serving. We propose a new LazyAR decoder architecture to relax layer-wise autoregressive dependencies and boost decoding throughput without hurting effectiveness. Beyond LLM-style optimizations, we further introduce recommendation-specific serving techniques, including Dynamic Beam Serving (DBS)—with Dynamic Beam Width (DBW) and Traffic-Aware Adaptive Beam Search (TABS)—as well as a short time-to-live (TTL) cache.

Online A/B tests show that GR4AD delivers up to 4.2\% ad revenue improvement compared to the existing DLRM-based stack, with consistent gains from both model scaling and inference-time scaling. With systematic efficiency optimizations in architecture and serving, GR4AD achieves <100ms latency and 500+ QPS per L20 under practical resource budgets, and has been fully deployed in Kuaishou advertising system serving over 400 million users.

\section{Related Works}
\subsection{Generative Recommendation}
Recent advances in generative models~\cite{gibbs,gan,dm2020,vae2019,gpt}, particularly Large Language Models (LLMs)~\cite{gpt}, have catalyzed a new paradigm in recommendation systems based on end-to-end generation. A representative line of work is Semantic ID–based generative recommendation~\cite{JuCNKW0S25}, which encodes items as discrete semantic identifiers (SIDs) and formulates recommendation as next-token prediction. TIGER~\cite{rajput2023tiger} is a seminal work that introduces hierarchical SIDs via residual quantization of item features and models recommendation as a generative retrieval task using an encoder–decoder Transformer. Building on this foundation, subsequent studies such as LC-Rec~\cite{zheng2024adapting}, LETTER~\cite{wang2024learnable}, and OpenOneRec~\cite{zhou2025openonerec} improve scalability and generalization by aligning item identifiers with collaborative semantics through two-stage training. OneRec~\cite{zhou2025onerec,zhou2025onerecv2} unifies retrieval and ranking within a single generative model and has shown strong performance in large-scale industrial recommendation, while GPR~\cite{zhang2025gpr} and OneSearch~\cite{chen2025onesearch} extend the generative paradigm to advertising and e-commerce search, respectively. To address the inference inefficiency of autoregressive generation, RPG~\cite{hou2025generating} and NEZHA~\cite{wang2025nezha} propose parallel and hyperspeed decoding mechanisms, and MMQ~\cite{xu2025mmq} further generalizes generative recommendation to multimodal settings via discrete quantization.

Despite these advances, existing generative recommenders largely lack architectures and learning strategies specifically designed for advertising systems, particularly under online learning constraints.

\subsection{Reinforcement Learning}

Reinforcement learning has become a central paradigm for LLM preference optimization. Early RLHF \cite{ouyang2022training} formulates alignment as policy optimization over a learned reward model, but its instability and engineering cost have motivated simpler preference-based alternatives. DPO \cite{rafailov2023direct} removes explicit reward modeling by optimizing directly over preference pairs, followed by variants such as SimPO \cite{meng2024simpo}, GRPO \cite{shao2024deepseekmath}, SAPO \cite{gao2025soft}, and GDPO \cite{liu2026gdpo}, which further improve stability or multi-objective learning. However, these methods are largely offline, relying on static preference data or fixed rollouts, and mainly address distributional alignment rather than continual learning in non-stationary environments.
In recommender, search, and advertising systems, reinforcement learning has been increasingly adopted to align large generative models with user behavior and business objectives. Works such as OneRec \cite{zhou2025onerec, zhou2025onerecv2}, OneSearch \cite{chen2025onesearch}, and GPR \cite{zhang2025gpr} adapt policy optimization via reward shaping, preference weighting, or hierarchical objectives, but often depend on multi-stage pipelines or naive combinations of SFT and RL, limiting adaptability to streaming settings. To address this, recent studies including SRFT \cite{fu2025srft}, CHORD \cite{zhang2025policy}, and the unified policy optimization framework in \cite{HPT} propose principled joint SFT–RSPO optimization through dynamic weighting or shared objectives. Building on these ideas, we design an efficient SFT–RSPO integration framework tailored to streaming recommender scenarios.
\subsection{Semantic IDs}
In traditional advertising systems, each item is assigned a sequentially increasing ID, which suffers from data sparsity and hampers cold-start. In contrast, Semantic IDs leverage LLMs to encode item content (text, video, etc.) into embeddings, which are then hierarchically quantized into structured identifiers. 
Quantization methods fall into two main categories: RQ-VAE-based and clustering-based. TIGER~\cite{rajput2023tiger} uses RQ-VAE to recursively quantize residuals for hierarchical IDs. FORCE~\cite{fu2025forge} introduces a cross-entropy equalization loss for balanced codebook usage. 
PLUM~\cite{he2025plum} uses multi-resolution codebooks and progressive masking for clear hierarchical structure. COST~\cite{zhu2024cost} proposes a contrastive learning framework to optimize the codebook for more discriminative IDs. QARM~\cite{luo2024qarm} adopts RQ‑Kmeans with cardinality constraints to dynamically regulate codebook usage, improving balance and efficiency. Furthermore, methods like ColaRec~\cite{wang2024content}, SEATER~\cite{si2024generative}, and TokenRec~\cite{qu2024tokenrec} integrate collaborative filtering: they first derive collaborative signal-enriched embeddings from user-item interactions, then encode them into discrete semantic IDs, injecting collaborative inductive bias into the generative recommendation framework.

\section{Methodology}

In this section, we detail the model architecture and training paradigm of \textbf{GR4AD}. Figure~\ref{fig:great} provides an overview of the method. 

\begin{figure*}[htbp]
\vspace{-1em}
        \centering
        \includegraphics[width=0.9\linewidth]{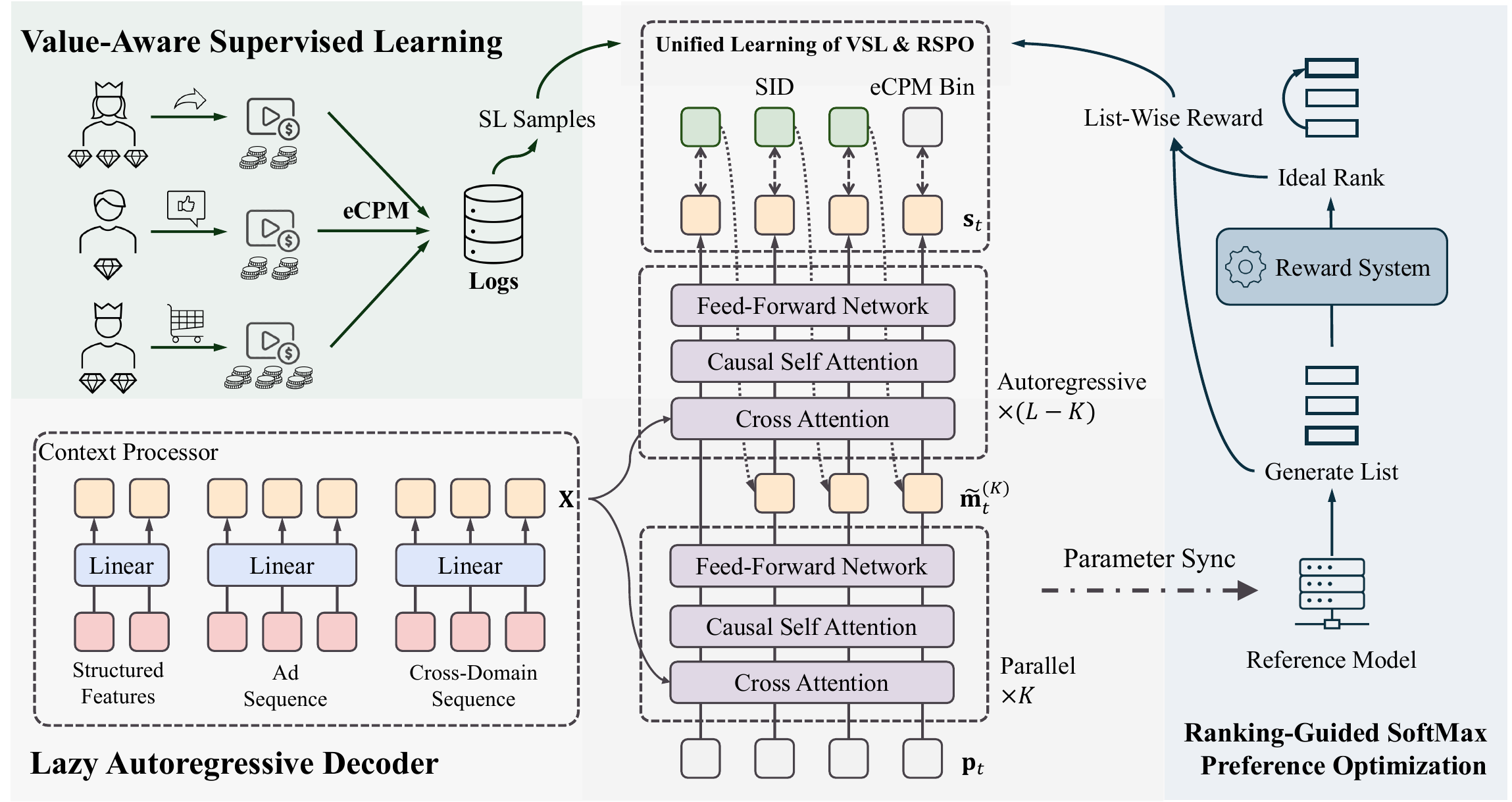}
\vspace{-1em}
    \caption{Overview of our proposed GR4AD: model architecture and learning algorithm.}
    \label{fig:great}
\vspace{-1em}
\end{figure*}

\subsection{Unified Advertisement Semantic ID}

\begin{figure}[htbp]
\vspace{-0.5em}
        \centering
        \includegraphics[width=0.95\linewidth]{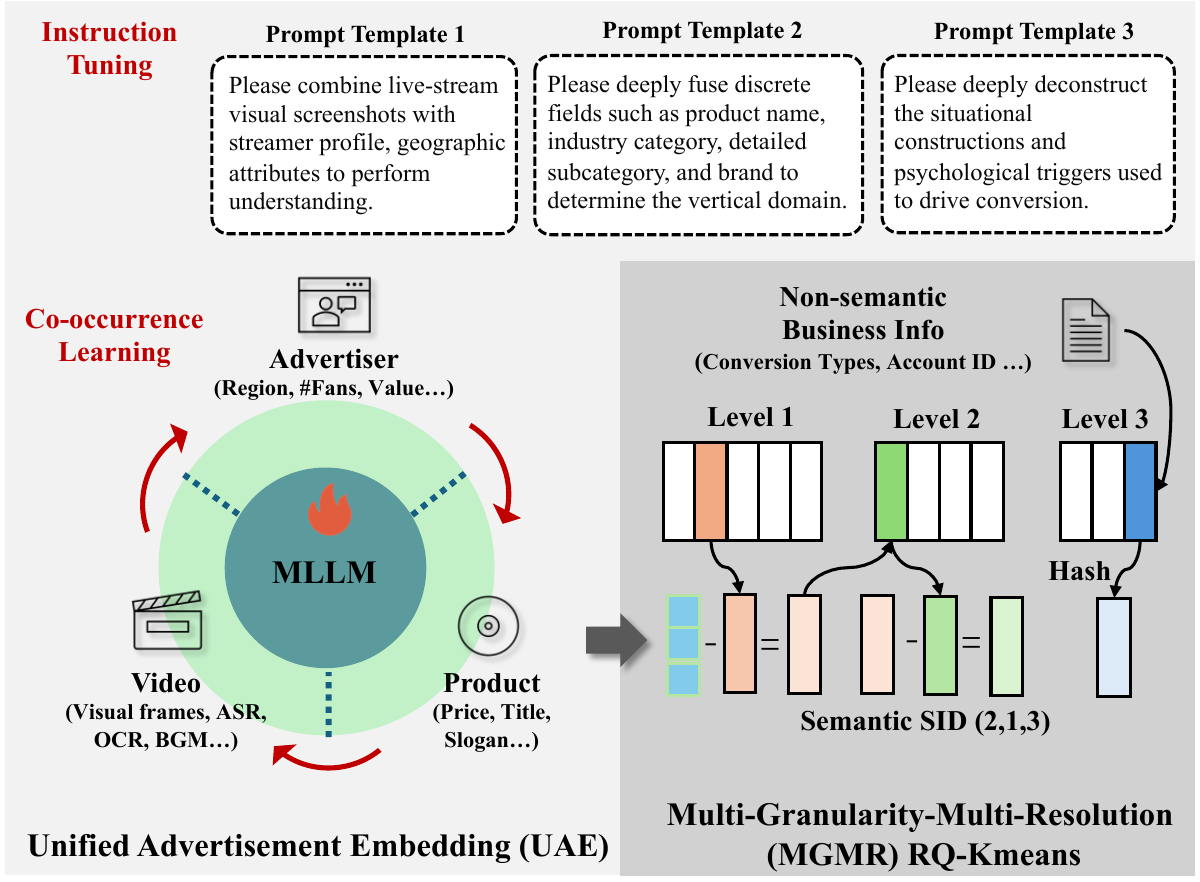}
\vspace{-0.5em}
    \caption{Illustration of Unified Advertisement Semantic ID. }
    \label{fig:ua-sid}
\vspace{-2em}
\end{figure}

In generative recommendation, Semantic IDs play a role analogous to tokens in large language models, where both representation granularity and quantization quality critically affect downstream performance. We therefore design \emph{Unified Advertisement Semantic IDs} (\textbf{UA-SID}) which is derived from an end-to-end finetuned advertisement LLM embedding model, as illustrated in the Figure~\ref{fig:ua-sid}.

\subsubsection{Unified Advertisement Embedding (UAE)}




 \paragraph{Instruction Tuning (IT)}
To capture heterogeneous content formats in advertising (video, product and advertiser), we employ instruction-based fine-tuning of an LLM to instill ad-understanding capabilities across diverse aspects and scenarios.
For example, for a live-stream host we prompt the model to analyze their profile and geographic information, while for ordinary off-platform advertisers we direct the model to focus on the product’s industry and brand information.
We designed 6 instructions in total to comprehensively cover Kuaishou’s ad types (presented in the appendix).


\paragraph{Co-occurrence Learning (CL)}

To incorporate collaborative signals from user behavior, we augment representation learning with a co-occurrence contrastive objective. Item co-occurrence strengths are estimated using the Swing method~\cite{yang2020large}. 
We further augment postive samples $\mathcal{P}_i$ with respect to each item triplets (Video, Product, Advertiser), where the co-occurring pair is treated as positives, while other in-batch samples serve as negatives. 
We adopt the InfoNCE:
\begin{equation}
\mathcal{L}_{\mathrm{NCE}}(i)
= - \log
\frac{
\sum\limits_{j \in \mathcal{P}_i}
\exp\!\left( \mathbf{z}_i^\top \mathbf{z}_j / \tau \right)
}{
\sum\limits_{k \neq i}
\exp\!\left( \mathbf{z}_i^\top \mathbf{z}_k / \tau \right)
},
\end{equation}
where $\mathbf{z}$ represents the last hidden states of MLLM.


\subsubsection{Multi-Granularity-Multi-Resolution (MGMR) RQ-Kmeans}

The stability and balance of the UA-SID codebook are critical for effective generative retrieval. 
Lower codebook utilization and higher SID collision ratio can degrade downstream performance.

Standard RQ-Kmeans with fixed codebook size often suffers from low codebook utilization~\cite{luo2024qarm}. 
Inspired by~\cite{he2025plum}, we adopt a \emph{Multi-Resolution (MR)} RQ-Kmeans scheme. To better preserve semantic separability, lower levels use larger codebooks to capture dominant factors early, while higher levels model lower-entropy residuals. 
Balanced K-means clustering is applied at each level to improve codebook utilization. 

SID collisions are also acute in advertising, identical-content ads can exhibit entirely different delivery trajectories if advertisers target distinct conversion types.  
In fact, ad systems routinely expose such \emph{Multi-Granularity (MG)}, predominantly numeric signals that lack conventional semantics. Therefore, we replace vector quantization at the final layer with a hash-based numeric mapping derived from non-semantic features (e.g., item/account IDs, conversion types), which markedly improves global balance and reduces collisions with negligible extra quantization error.


Finally, each item is mapped to a discrete UA-SID sequence
\begin{equation}
\mathbf{y} = (s_{1}, s_{2}, \ldots, s_{T}), \quad s_t \in \mathcal{V}_{t},
\end{equation}
where $s_t$ denotes the token at level $t$, $T$ is the UA-SID depth (typically small), and $\mathcal{V}_{t}$ is the vocabulary at level $t$.


\subsection{Lazy Autoregressive Decoder}
\label{sec:lazyar}

For structured features and user interaction sequences, we follow LazyDecoder~\cite{zhou2025onerecv2} and adopt a lightweight linear context processor to efficiently model these heterogeneous contexts, whose output is denoted as $\mathbf{X} = (\mathbf{x}_1, \ldots, \mathbf{x}_S)$ with $\mathbf{x}_i \in \mathbb{R}^{d}$, where $S$ is the context length and $d$ is the hidden size. The decoder then generates the target UA-SID $\mathbf{y}$.

\paragraph{Vanilla Autoregressive Decoding.}
Prior work commonly uses a standard autoregressive decoder that factorizes
\begin{equation}
p(\mathbf{y} \mid \mathbf{X}) = \prod_{t=1}^{T} p\!\left(s_t \mid s_{<t}, \mathbf{X}\right),
\label{eq:ar_factor}
\end{equation}
and feeds the embedding of the previous-level UA-SID back to the \emph{first} decoder layer at the next decoding step, as shown in Figure~\ref{fig:lazyar}(a). Concretely, let $\mathbf{s}_{t} \in \mathbb{R}^d$ be the embedding of $s_{t}$ (and $\mathbf{s}_{0}=\texttt{BOS}$) and $\mathbf{p}_t$ be the position embedding of step $t$. A naive decoder initializes the step-$t$ state with $\mathbf{h}^{(0)}_t=\mathbf{s}_{t-1}+\mathbf{p}_t$ and applies $L$ decoder layers:
\begin{equation}
\mathbf{h}^{(\ell)}_t = \mathrm{Dec}^{(\ell)}\!\left(\mathbf{h}^{(\ell-1)}_t, \mathbf{X}\right), 
\quad \ell=1,\ldots,L,
\end{equation}
followed by a classifier $p(s_t\mid \cdot)=\mathrm{Softmax}(\mathbf{W}_t\mathbf{h}^{(L)}_t)$.

\paragraph{LazyAR: Late-Inject Autoregression.}
In the experiments, we observed that the first-level UA-SID ($s_1$) typically has the largest loss and is the most important to learn, yet contributes little to beam-search cost: decoding starts from \texttt{BOS}, so the effective beam for $s_1$ is $1$, while beams become much larger for later levels ($t \gg 1$). As a result, most decoding compute is spent on later levels, which are empirically easier. Motivated by this mismatch between learning difficulty and inference cost, we ask whether we can reduce the decoding compute for later-level UA-SID without affecting the inference of the first level.

A straightforward approach is to add extra shallow decoders for later UA-SID levels (e.g., DeepSeek MTP~\cite{deepseek2024dv3}, shown in Figure~\ref{fig:lazyar}(b)), but this (i) introduces additional parameters, (ii) prevents early decoder layers from directly participating in later-level inference and (iii) empirically decreases the recommendation performance. Instead, we propose \textbf{LazyAR} (\textbf{L}azy \textbf{A}uto\textbf{R}egression), which delays the dependence on $s_{t-1}$ to an intermediate layer. Detailedly, given a chosen $K$ ($1 \le K < L$), for each UA-SID level $t$, LazyAR computes the first $K$ layers \emph{without} conditioning on $s_{t-1}$:
\begin{align}
\mathbf{m}^{(0)}_t &= \mathbf{p}_t, \\
\mathbf{m}^{(\ell)}_t &= \mathrm{Dec}^{(\ell)}\!\left(\mathbf{m}^{(\ell-1)}_t, \mathbf{X}\right),
\quad \ell=1,\ldots,K,
\label{eq:trunk}
\end{align}
then injects the previous-level UA-SID embedding at layer $K$ via a fusion operator:
\begin{equation}
\tilde{\mathbf{m}}^{(K)}_t = \mathrm{Fuse}\!\left(\mathbf{m}^{(K)}_t, \mathbf{s}_{t-1}\right),
\label{eq:fuse}
\end{equation}
and applies the remaining $L-K$ layers autoregressively:
\begin{equation}
\mathbf{h}^{(\ell)}_t = \mathrm{Dec}^{(\ell)}\!\left(\mathbf{h}^{(\ell-1)}_t, \mathbf{X}\right),
\quad \ell=K+1,\ldots,L,
\label{eq:head}
\end{equation}
with $\mathbf{h}^{(K)}_t \triangleq \tilde{\mathbf{m}}^{(K)}_t$ and
\begin{equation}
p(s_t \mid s_{<t}, \mathbf{X}) = \mathrm{Softmax}\!\left(\mathbf{W}_t\mathbf{h}^{(L)}_t\right).
\end{equation}

We implement $\mathrm{Fuse}(\cdot)$ as a lightweight gated projection:
\begin{equation}
\mathrm{Fuse}(\mathbf{m},\mathbf{s}) = \mathbf{W}_f[\mathbf{m}\odot(\mathbf{W}_g\mathbf{s});\mathbf{s}],
\end{equation}
where $[\cdot;\cdot]$ is concatenation and $\odot$ is element-wise product. Each decoder layer $\mathrm{Dec}^{(\ell)}(\cdot)$ follows OneRecV2~\cite{zhou2025onerecv2}, consisting of a cross-attention between $\mathbf{m}^{(\ell)}_t$/$\mathbf{h}^{(\ell)}_t$ and $\mathbf{X}$, a self-attention module over decoder states, and a feed-forward network. All submodules are equipped with pre-layer normalization to stabilize training and facilitate deeper architectures.


\begin{figure}[tbp]
\vspace{-0.5em}
        \centering
        \includegraphics[width=1.0\linewidth]{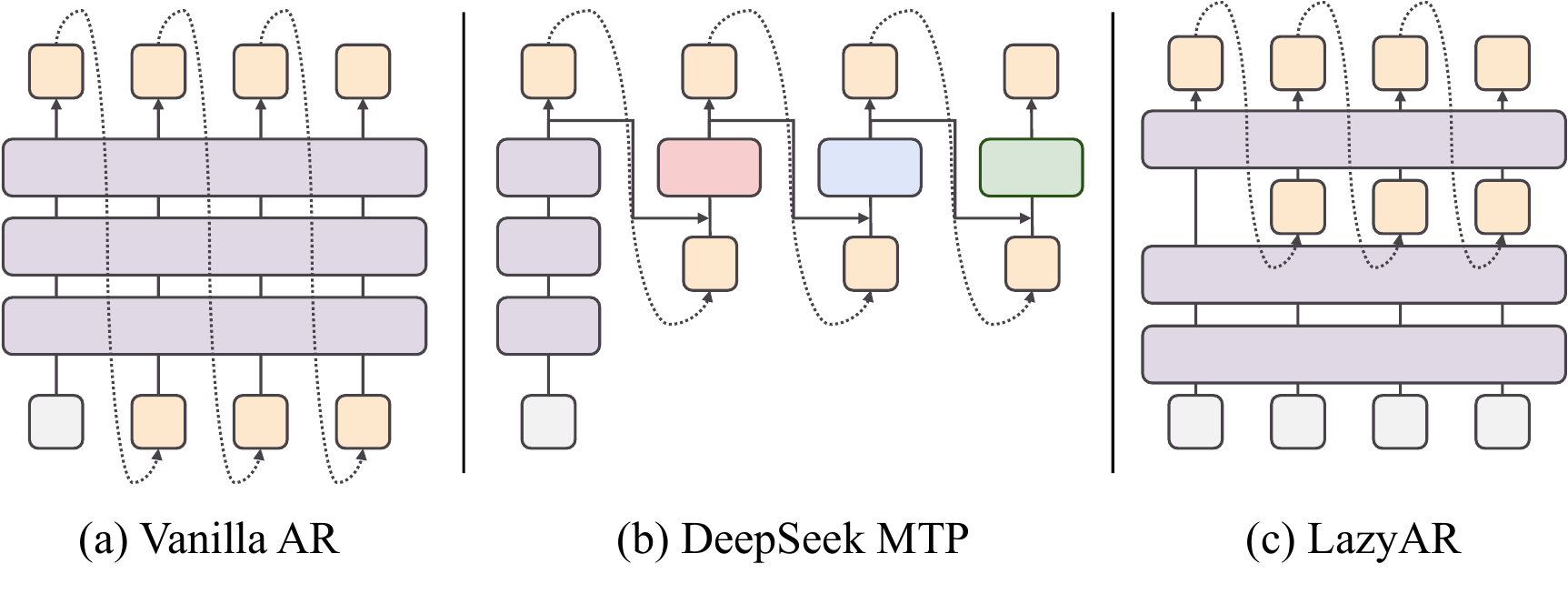}
\vspace{-2.5em}
    \caption{Comparison of Vanilla AR, DeepSeek MTP, LazyAR.}
    \label{fig:lazyar}
\vspace{-1em}
\end{figure}

\paragraph{Why LazyAR is faster.}
$\mathbf{m}^{(K)}_t$ does not depend on $s_{t-1}$, so the first $K$ layers can be computed in parallel for all levels and reused across beams:
\begin{equation}
\{\mathbf{m}^{(K)}_t\}_{t=1}^{T} = \mathrm{Dec}^{(1:K)}\!\left(\{\mathbf{p}_t\}_{t=1}^{T}, \mathbf{X}\right).
\end{equation}

Only the remaining $L-K$ layers (Eq.~\eqref{eq:head}) require autoregressive dependency, reducing sequential work per level during beam search. This is important because, at later UA-SID levels, decoding must be performed for a large number of beam candidates in parallel, making the cost of the autoregressive component a major bottleneck in overall decoding efficiency.

\paragraph{Why LazyAR preserves performance.}
(1) The decoding process for the first-level UA-SID remains unchanged: $s_1$ is still generated from a initialized embedding through the full $L$ decoder layers, ensuring no degradation on the hardest level.
(2) For later steps, LazyAR directly leverages the representations produced by the first $K$ layers. These intermediate states are not arbitrary -- they can reason in latent space and encode useful signals about plausible candidates at level $t$. 
To encourage the first $K$ layers to learn richer and more predictive latent representations, we introduce an MTP-style auxiliary loss. Concretely, we bypass the fusion projection and set $\mathbf{h}^{(K)}_t \triangleq \mathbf{m}^{(K)}_t$, forcing the trunk to provide sufficient information for downstream decoding even without explicit dependence on the previous-level UA-SID. 
In DeepSeek-MTP, both inputs to the fusion layer contain only information from the previous token, and the earlier decoder layers are used solely to support decoding at the first step.
(3) $K$ is a tunable hyper-parameter that provides a flexible accuracy-efficiency trade-off without changing the model size. When $K=1$, LazyAR reduces to naive autoregression. This design differs substantially from DeepSeek-MTP, where adding decoder layers for subsequent steps results in a substantial growth in model parameters. In our experiments, $K=\frac{2}{3}L$ preserves recommendation quality while doubling inference throughput. 
(4) Furthermore, because the first $K$ layers are computed only once and shared across beams, their impact on beam-search inference cost is significantly reduced. It suggests that these $K$ layers can potentially be extended to perform more inference steps and incorporate latent reasoning, opening up additional capacity for improving model effectiveness. We leave this exploration to future work.

\paragraph{Discussion}
This design is recommendation-specific and is less suitable for standard LLM decoding. In typical LLM serving, beam search is often not used (or uses a small beam), and the difficulty of predicting later tokens does not necessarily decrease, so deferring autoregressive dependency to later layers may yield limited speedup and offers no guarantee for long, variable-length generations. 

\subsection{Value-Aware Supervised Learning}

We train GR4AD using a Value-Aware Supervised Learning (VSL) objective tailored to advertising recommendation. Similar to language models, the core supervision signal is defined on discrete tokens and optimized via cross-entropy loss.
Given the UA-SID sequence $\mathbf{y}$, we first define a standard autoregressive token prediction loss over SID tokens:
\begin{equation}
\mathcal{L}_{\mathrm{SID}}
= - \sum_{t=1}^{T} \log p(s_t \mid s_{<t}, \mathbf{X}).
\end{equation}

\paragraph{ECPM-Aware Token Prediction.}
To better adapt the training objective to advertising scenarios, we further incorporate business value information by introducing an \emph{eCPM token}, which can be used to re-rank the generated SIDs. Specifically, we discretize the continuous eCPM values of training samples into equiprobable buckets and append the resulting eCPM token as an additional prediction step after the UA-SID sequence. The eCPM prediction is optimized using a cross-entropy loss:
\begin{equation}
\mathcal{L}_{\mathrm{eCPM}}
= - \log p(v \mid \mathbf{y}, \mathbf{X}),
\end{equation}
where $v$ denotes the discretized eCPM token. The combined next-token prediction loss is then
\begin{equation}
\mathcal{L}_{\mathrm{NTP}} = \mathcal{L}_{\mathrm{SID}} + \lambda_e\mathcal{L}_{\mathrm{eCPM}}.
\end{equation}

\paragraph{Value-Aware Sample Weighting.}
In advertising, training samples exhibit highly skewed value distributions. To reflect their heterogeneous importance, we apply a value-aware weighting scheme to all loss terms. Each sample is assigned a weight $w = w_{\mathrm{user}} \cdot w_{\mathrm{behavior}}$,
where $w_{\mathrm{user}}$ captures the long-term advertising value of the user, and $w_{\mathrm{behavior}}$ reflects the depth of user interaction (e.g., purchase actions receive higher weights than clicks). Samples from users with higher advertising value and deeper engagement thus contribute more strongly during training.

\paragraph{Auxiliary MTP Loss.}
As discussed in Section~\ref{sec:lazyar}, LazyAR computes the first $K$ decoder layers without conditioning on the previous token. To encourage these parallel layers to learn richer and more predictive representations, we introduce an auxiliary multi-token prediction (MTP) loss by setting $\mathbf{h}^{(K)}_t \triangleq \mathbf{m}^{(K)}_t$. Concretely, we require the trunk states to directly predict target tokens without relying on late-injected autoregressive signals. This auxiliary loss is applied during training only and is denoted as $\mathcal{L}_{\mathrm{MTP}}$.

The final VSL objective is defined as
\begin{equation}
\mathcal{L}_{\mathrm{VSL}}
= \mathbb{E}_{\mathcal{D}} \left[
w \left(
\mathcal{L}_{\mathrm{NTP}} + \lambda_{\mathrm{mtp}} \mathcal{L}_{\mathrm{MTP}}
\right)
\right],
\end{equation}
where $\lambda_{\mathrm{mtp}}$ controls the strength of the auxiliary MTP loss.

\subsection{Ranking-Guided Reinforcement Learning}

VSL enables GR4AD to learn a distribution over UA-SIDs that reflects historical user interests and advertising signals, but it mainly fits the logged data distribution and does not directly optimize downstream objectives. In particular, it does not explicitly favor UA-SIDs with higher advertising value or support exploration beyond observed behaviors. We therefore introduce a ranking-guided RL stage on top of VSL to encourage value-aware, list-level optimization while enabling controlled exploration, remaining grounded in the learned distribution.

\subsubsection{RSPO}
Unlike LLMs, recommendation systems aim to generate a ranked list rather than a single output, making per-item rewards insufficient to fully capture list-level optimization. Moreover, training samples in our setting are not limited to log feedback generated by the generative model itself, but also include samples collected from other production pipelines.

To address these challenges, we propose a list-wise RL method tailored to advertising recommendation, termed \textbf{RSPO} (Ranking-Guided Softmax Preference Optimization). Let $\mathcal{Y} = \{\mathbf{y}_1, \mathbf{y}_2, \ldots, \mathbf{y}_n\}$ denote the candidate list, and $v_i$ denotes the associated reward (eCPM) of $\mathbf{y}_i$. Inspired by Lambda framework~\cite{lambdaloss} and SDPO~\cite{SDPO}, instead of constructing chosen–rejected pairs using heuristic rules, RSPO directly aligns the RL objective with the ranking NDCG:
\begin{align}
\mathcal{L}_\text{RSPO}
&
    \begin{aligned}[t]
        & = -\mathbb{E}_{(X, \mathbf{y}_i, \mathcal{E}_i) \sim \mathcal{D}}\Bigg[ \log_{2} \sigma \Bigg( -\log \sum_{\mathbf{y}_j \in \mathcal{E}_i} \mathcal{M}_{ij} \exp \\
        &\quad \Bigg(\beta \log \frac{p_\theta(\mathbf{y}_j \mid X)}{{p_{ref}(\mathbf{y}_j \mid X)}^{C_{ij}}} - \beta \log \frac{p_\theta(\mathbf{y}_i \mid X)}{{p_{ref}(\mathbf{y}_i \mid X)}^{C_{ij}}} 
     \Bigg) \Bigg].
    \end{aligned}
\end{align}

Here, $\mathcal{E}_i = \{\mathbf{y}_j \mid v_j < v_i\} \subset \mathcal{Y}$ denotes the set of candidates ranked below $\mathbf{y}_i$. The coefficient $\mathcal{M}_{ij} = \left| \frac{1}{D_{|i-j|}} - \frac{1}{D_{|i-j|+1}} \right| |G_i - G_j|$ follows the standard Lambda formulation, where $G_i = \frac{2^{v_i} - 1}{Z}$, $D_i = \log_{2}(1 + i)$, and $Z$ is the ideal DCG. We show that $\mathcal{L}_\text{RSPO}$ is an upper bound of NDCGcost (proof in Appendix~\ref{app:rspo}) with respect to the ranking induced by 
$\log \frac{p_\theta(\mathbf{y}_i \mid X)}{{p_{ref}(\mathbf{y}_i \mid X)}^{C_{ij}}}$:
\begin{equation}
\begin{aligned}
\mathrm{NDCGcost} = \sum_{i=1}^{n} G_i - \sum_{i=1}^{n} \frac{G_i}{D_i} = \sum_{i=1}^{n} G_i - \mathrm{NDCG},
\end{aligned}
\end{equation}
$\beta$ controls the preference strength and $C_{ij}$ is a binary gate for reference availability and reliability. In production, training samples come from heterogeneous sources. Some lists are generated by GR4AD, for which we can record historical online predictions as a reference distribution $p_{ref}$; however, many samples are collected from other pipelines and thus have no reliable $p_{ref}$. Moreover, even for GR4AD-generated samples, $p_{ref}$ may become stale due to distribution drift and training--serving inconsistencies. When the current model predictions deviate too much from $p_{ref}$, enforcing the reference constraint can introduce noisy regularization and lead to unstable updates. We therefore enable $p_{ref}$ only when it is available and deemed reliable; otherwise, we drop the reference:
\begin{align}
C_{ij} = 
\begin{cases}
1 & \frac{1}{|\mathcal{E}_i \cup \{\mathbf{y}_i\}|} \sum_{\mathbf{y}_t \in \mathcal{E}_i \cup \{\mathbf{y}_i\}} \left| \log \frac{p_{\theta}(\mathbf{y}_t \mid X)}{p_{\text{ref}}(\mathbf{y}_t \mid X)} \right| < \delta \\
0 & \text{otherwise}
\end{cases},
\end{align}
where $\delta$ is a hyperparameter threshold.

\subsubsection{Unified Learning of VSL and RSPO}

In LLM training, pre-training, supervised fine-tuning, and reinforcement learning are typically conducted in separate stages. In our production setting, however, GR4AD is updated continuously via online learning, making it essential to \emph{jointly} integrate VSL and RSPO in a single training stream. Inspired by HPT~\cite{HPT}, we treat VSL as learning a stable \emph{base distribution} over user-interest items, while RSPO refines this distribution by biasing generation toward \emph{higher-value} items without drifting away from user relevance.

To balance imitation and exploration dynamically, we introduce a sample-level \emph{alignment score} that measures the mismatch between the model's current preference and the reward signal. For a candidate list of size $n$, let $r_p^{(i)}$ and $r_v^{(i)}$ be the ranks of candidate $i$ according to the model likelihood $p_{\theta}(\mathbf{y}_i \mid X)$ and its reward $v_i$ (eCPM), respectively. We define the normalized rank discrepancy:
\begin{equation}
A^{(i)} = \frac{|r_p^{(i)} - r_v^{(i)}|}{n-1}, \quad A^{(i)} \in [0,1].
\end{equation}

When $A^{(i)}$ is large, the model ranking deviates from the reward ranking, indicating insufficient imitation of the user-interest distribution; thus we increase the weight of VSL. When $A^{(i)}$ is small, the model is already broadly aligned, and we place more emphasis on RSPO to improve value-aware listwise optimization.

We assign per-sample weights for the two objectives as:
\begin{align}
w_{\mathrm{VSL}}^{(i)} &= w_0 \cdot \exp\!\Big(A^{(i)} \cdot \log(1+v_i)\Big), \\
w_{\mathrm{RL}}^{(i)}  &= w_0 \cdot Z_{\max}(1-A^{(i)}),
\end{align}
where $w_0$ is a base scaling factor and $Z_{\max}$ caps the RL weight.
The final unified online training objective is
\begin{equation}
\mathcal{L} = \mathbb{E}_{i \sim \mathcal{D}} \left[
w_{\mathrm{VSL}}^{(i)} \, \mathcal{L}_{\mathrm{VSL}}^{(i)} \;+\;
w_{\mathrm{RL}}^{(i)} \, \mathcal{L}_{\mathrm{RSPO}}^{(i)}
\right].
\end{equation}

\begin{figure}[tbp]
        \centering
        \includegraphics[width=1.0\linewidth]{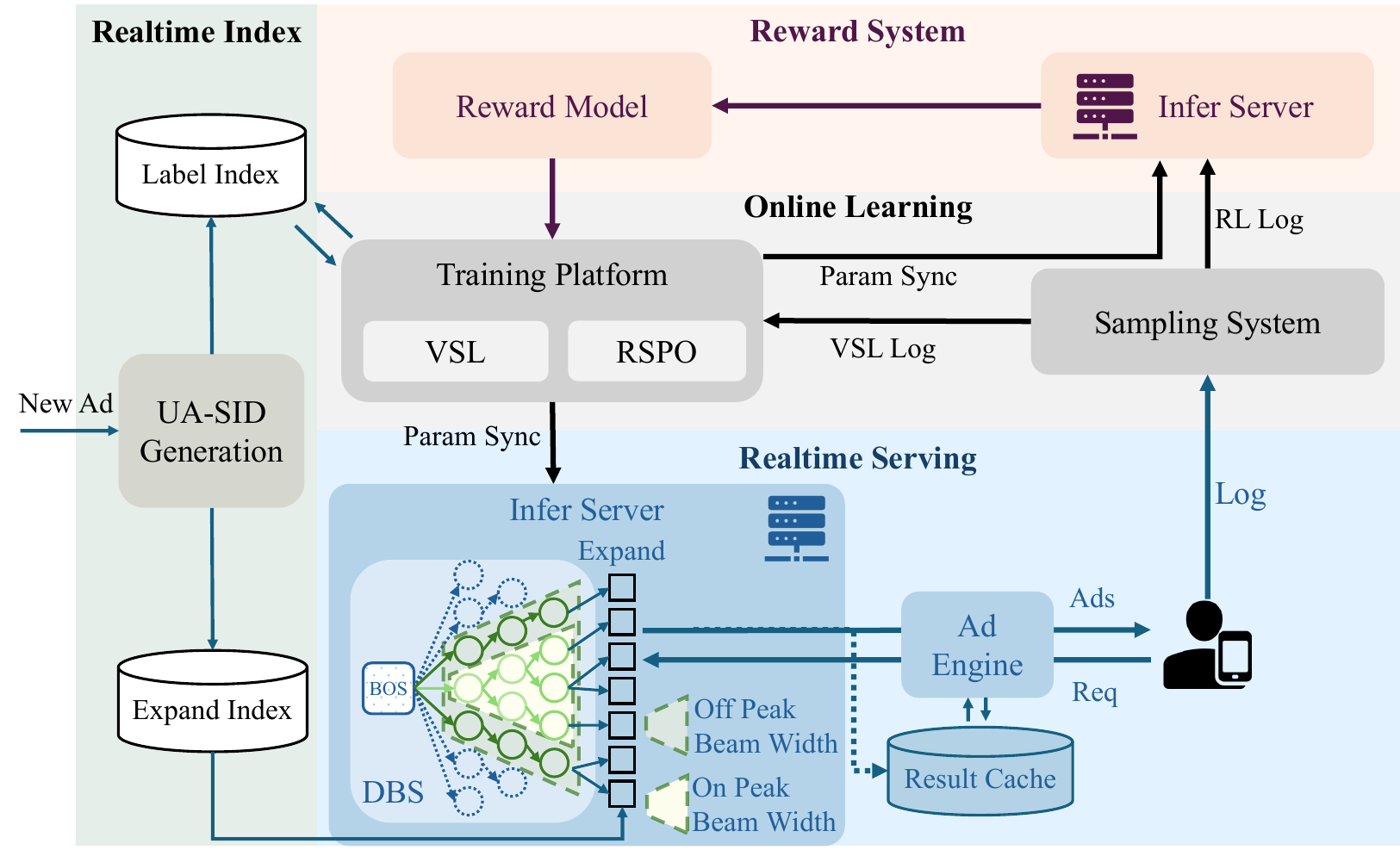}
\vspace{-2em}
    \caption{System overview: training and serving of GR4AD.}
    \label{fig:system}
\vspace{-1em}
\end{figure}

\section{Deployment}

We have deployed GR4AD (0.16B) in Kuaishou’s advertising system, serving over 400 million users. The system employs a closed-loop architecture integrating reward estimation, online learning, and real-time indexing for continuous model evolution, while specific efficiency optimizations ensure high-throughput (500+ QPS per L20), low-latency (<100ms) inference under fluctuating traffic.

\subsection{System Overview}
As shown in Figure~\ref{fig:system}, GR4AD operates as a closed-loop system with four components: a \textit{Realtime Serving} engine for request handling and ranking, a \textit{Realtime Index} module for Item-SID mapping, an \textit{Online Learning} module for continual VSL/RL updates, and a \textit{Reward System} that provides value-based feedback to close the loop.

\paragraph{Reward System}
The Reward System supplements online logs for RL training, since relying only on online exposed SL samples is (i) data-limited and insufficient to align optimization with revenue objectives (e.g., eCPM), and (ii) directly performing exploration in real-time serving can incur noticeable performance cost. We train a reward model on real exposure data and use it to score candidate sets generated by GR4AD. In the Reward System, latency constraints are relaxed, allowing GR4AD to generate more candidates with a larger beam and to introduce controlled random exploration. The resulting value estimates (e.g., eCPM) are streamed as RL logs to support ranking-aware RSPO updates.

\paragraph{Online Learning module}
The Online Learning Module processes real-time request and interaction streams to construct two distinct training signals: (i) it prioritizes positive engagement samples to generate VSL logs for modeling the user-interest distribution, and (ii) it samples a subset of high-value users—including cases where explicit interactions may be absent—and routes their contexts to the reward system for value estimation. The online learner performs continual mini-batch training for both VSL and RL, pushing updated parameters to the inference server in real time.

\paragraph{Realtime Index module}
The Realtime Index Module replaces the legacy embedding-based retrieval pipeline to enable generative recommendation. In traditional DLRM systems, retrieval relies on embedding indexes that must be frequently rebuilt as models update, making index refresh typically minute-level. In contrast, our SID-based indexing is driven by a relatively stable content-to-UA-SID mapping and does not require frequent global rebuilds. When a new item arrives, we simply compute its UA-SID from content and update a bidirectional index (UA-SID $\leftrightarrow$ Item ID) in seconds. This design substantially improves item freshness and cold-start coverage, while also ensuring training--serving consistency.

\paragraph{Realtime Serving}
Realtime Serving handles user requests by invoking the inference server to return ranked ad lists, while simultaneously logging serving contexts and user feedback to close the loop. To maintain high-throughput, low-latency inference under traffic fluctuations, we apply Dynamic Beam Serving (DBS), result caching, and other concurrency optimizations.


\subsection{Efficiency Optimization}
To reconcile the high computational demands of generative recommendation with strict production latency constraints, we introduce a suite of optimizations.

\subsubsection{Dynamic Beam Serving (DBS)}
We propose Dynamic Beam Serving (DBS) to improve the efficiency--effectiveness trade-off of multi-step decoding in real-time advertising. 

\paragraph{Dynamic Beam Width (DBW)}
The number of returned candidates is determined by the beam width at the \emph{final} step, while the overall decoding cost is dominated by the beam widths in \emph{earlier} steps, since they control the number of hypotheses propagated to subsequent steps. Therefore, a fixed beam width across steps is often suboptimal under a constrained compute budget. Motivated by this, DBS adopts a \emph{progressively increasing} beam schedule across steps, which reduces intermediate computation while preserving the final candidate quality.


\paragraph{Traffic-Aware Adaptive Beam Search (TABS)}
Moreover, request traffic in recommendation systems exhibits strong peak--off-peak cycles, and serving constraints are dictated by peak load. DBS therefore further adjusts the overall beam scale according to instantaneous traffic. Let $Q_t$ denote the serving QPS at time $t$ and $B_{\text{base}}$ be the baseline beam setting. We adapt the active beam scale based on traffic intensity and available computational slack $\mathcal{C}_{\text{avail}}$ by $B_t = B_{\text{base}} \cdot f(Q_t, \mathcal{C}_{\text{avail}})$. During off-peak periods (e.g., $Q_t < Q_{\text{threshold}}$), we increase $B_t$ to leverage otherwise idle compute, enabling broader hypothesis exploration and improving ranking quality, while keeping peak-time latency and throughput within budget.



%
\subsubsection{Reco Result Cache}
A single user may issue multiple ad requests within a short time window, during which both user intent and the candidate ad pool typically remain stable. As a result, inference outcomes are expected to be largely consistent across these requests. By caching previously generated ad recommendations, subsequent requests within a bounded interval (e.g., one minute) can directly reuse cached results, significantly reducing inference resource consumption without degrading serving performance.

\subsubsection{Other Optimizations.}
We propose Beam-Shared KV Caching to organize beams along the sequence dimension. This allows multiple beams to share a single encoder KV cache, eliminating redundant memory accesses and reducing the per‑step KV read complexity from $\mathcal{O}(B \cdot L)$ to $\mathcal{O}(L)$. For beam search, we introduce TopK Pre‑Cut. It first selects $k$ candidates in parallel from each beam of the previous step, then performs a global top‑$k$ selection over the aggregated candidates. This reduces the search space, improves GPU parallelism, and maintains search correctness. Further, we reduce numerical precision from FP32 to FP8~\cite{FP8}, significantly lowering both computational redundancy and memory access overhead.



\begin{table}[t]
\centering
\caption{Overall performance and ablation of GR4AD.}
\vspace{-1em}
\resizebox{0.95\columnwidth}{!}{%
\begin{tabular}{l c c}
    \toprule
    \multirow{2}{*}{\textbf{Model Settings}}                                           & {\textbf{$\Delta$Revenue}}  & \textbf{$\Delta$QPS}     \\
                                                                      & vs. Base                    & vs. GR-Base              \\
                                                               
    \midrule
    \textit{Baselines}                                                                                                        \\
    \ \ \ \ DLRM (Base)                                               & --                          & --                      \\
    \ \ \ \ OneRec-V2~\cite{zhou2025onerecv2} (GR-Base)                                       & +1.68\%                     & --                      \\
    \midrule
    \textit{Tokenization Optimizations}                                                                                \\
    \ \ \ \ + UA-SID                                                     & +1.92\%                     & 0\%                  \\
    \midrule                                                                                         
    \textit{Learning Optimizations}                                                                                \\
    \ \ \ \ + VSL                                                     & +2.80\%                     & -25\%                   \\
    \ \ \ \ + VSL + DPO~\cite{rafailov2023direct}                     & +3.16\%                     & -25\%                   \\
    \ \ \ \ + VSL + GRPO~\cite{shao2024deepseekmath}                  & +3.21\%                     & -25\%                   \\
    \ \ \ \ + VSL + RSPO                                              & +3.86\%                     & -25\%                   \\
    \ \ \ \ + \textbf{U}nified \textbf{V}SL \& \textbf{R}SPO (UVR)    & +4.01\%                     & -25\%                   \\
    \midrule
    \textit{Serving Optimizations}                                                                                                          \\
    \ \ \ \ + UVR + DBS                                               & \textbf{+4.32\% }           & +20\%                   \\
    \ \ \ \ + UVR + DBS + DeepSeek-MTP~\cite{deepseek2024dv3}         & +3.98\%                     & \textbf{+117}\%         \\
    \ \ \ \ \textbf{GR4AD} (+ UVR + DBS + LazyAR)                     & +4.28\%                     & \textbf{+117}\%         \\
    \bottomrule
\end{tabular}
}
\label{tab:ablation_main}
\vspace{-1em}
\end{table}

\section{Experiments}


%
%
\subsection{Overall Performance}
%
As shown in Table~\ref{tab:ablation_main}, GR4AD significantly outperforms these baselines in both inference efficiency and revenue, demonstrating its effectiveness. Specifically, the baselines include previous DLRM-based Kuaishou advertising platform, and OneRec-V2\cite{zhou2025onerecv2}, a state-of-the-art generative recommendation model. Note that GR models are typically served on more powerful GPUs and with higher device utilization, whereas the DLRM-based production stack involves multiple models co-serving, making single-model QPS not directly comparable. Therefore, we report serving efficiency as relative QPS improvements over the GR-Base setting for fair comparison.
Moreover, we conduct several extensive ablation studies to assess the contribution of each component in GR4AD.
\paragraph{Value-Aware Online Learning}
%

%
(1) First, we find that introducing VSL significantly boosts online revenue, demonstrating the effectiveness of the VSL module. Specifically, VSL, incorporating user and behavior weighting as well as eCPM token prediction, better aligns with advertising scenario requirements by enabling differentiated user modeling and directly targeting business objectives to maximize revenue.
(2) However, while VSL learns a fixed data distribution in a point-wise manner, it lacks strong generalization ability and fails to further explore user preferences. To address this limitation, incorporating RSPO aligns generation probabilities with the relative ranking order in a list-wise manner, resulting in the most significant improvement among all optimization components. This demonstrates that RSPO can more comprehensively capture user interest and preferences compared to DPO\cite{rafailov2023direct} or GRPO\cite{shao2024deepseekmath}, leading to enhanced business revenue.
(3) Finally, rather than simply combining VSL and RSPO, we unify them through a sample-level training indicator, effectively leveraging the strengths of both. This approach not only stabilizes training in an online learning setting but also results in additional revenue.
\paragraph{Dynamic Beam Serving}
As shown in Table~\ref{tab:ablation_main}, we implement the DBS mechanism to optimize inference efficiency and maximize revenue. First, DBW allows us to replace a fixed beam width across layers with a dynamic beam width (e.g., replacing 512-512-512 with 128-256-512), which reduces the overall computational load, significantly improves inference efficiency, and does not compromise revenue. Then, with the help of TABS, we increase the beam width by 60\% during off-peak periods to improve revenue, while keeping the beam width unchanged during peak periods. 
The collaboration of both mechanisms leads to an optimized balance between revenue and efficiency.

\paragraph{Lazy Autoregressive Decoder}
We observe that LazyAR results in a marginal performance decrease, yet it nearly doubles qps. 
This highlights the dual benefits of LazyAR: on the one hand, it improves inference efficiency by sharing the majority of decoder layers and enabling parallel computation to reduce the overall computational load, and on the other hand, it introduces an MTP-style auxiliary loss to enrich latent representations, ensuring that performance remains as unaffected as possible. In LazyAR, we configure the total number of decoder layers, L=9, with the first K = 6 layers shared across all beams. This favorable accuracy-efficiency trade-off allows GR4AD to effectively handle Kuaishou's full advertising traffic.
\paragraph{Business Indicators}
Beyond revenue, there is a significant 17.5\% increase in ad delivery for small and medium-sized advertisers, marking a remarkable achievement. Additionally, due to more precise interest modeling, the generative ad system enhanced the user experience, leading to a 10.17\% improvement in ad conversion rates. Among less active users, we still observed a 7.28\% increase in conversion rates. We attribute these gains to improved generalization from content-based SIDs and a more real-time index that better supports cold-start items. Overall, GR4AD supports a healthier platform ecosystem, delivering a win-win-win outcome for the platform, advertisers, and users.
\subsection{Scaling Laws for GR4AD}

To validate the scalability of our approach, we further investigate scaling effects from both model parameters and inference beam.

\subsubsection{Model Scaling.}
We conduct a controlled set of online A/B tests across four model sizes, containing 0.03B, 0.08B, 0.16B, and 0.32B parameters, while keeping the inference beam width fixed at 512. We observe a clear and monotonic improvement in the revenue metric as model size increases. Larger models consistently achieve lower training loss, which is consistent with the observed online performance trend and suggests stronger representational power and generative modeling capacity. Specifically, the revenue lift steadily increases from +2.13\% to +4.43\%, demonstrating that scaling generative recommenders yields substantial and reliable gains in real-world advertising settings. 

    

    
    
\begin{figure}
    \centering
    \includegraphics[width=0.95\linewidth]{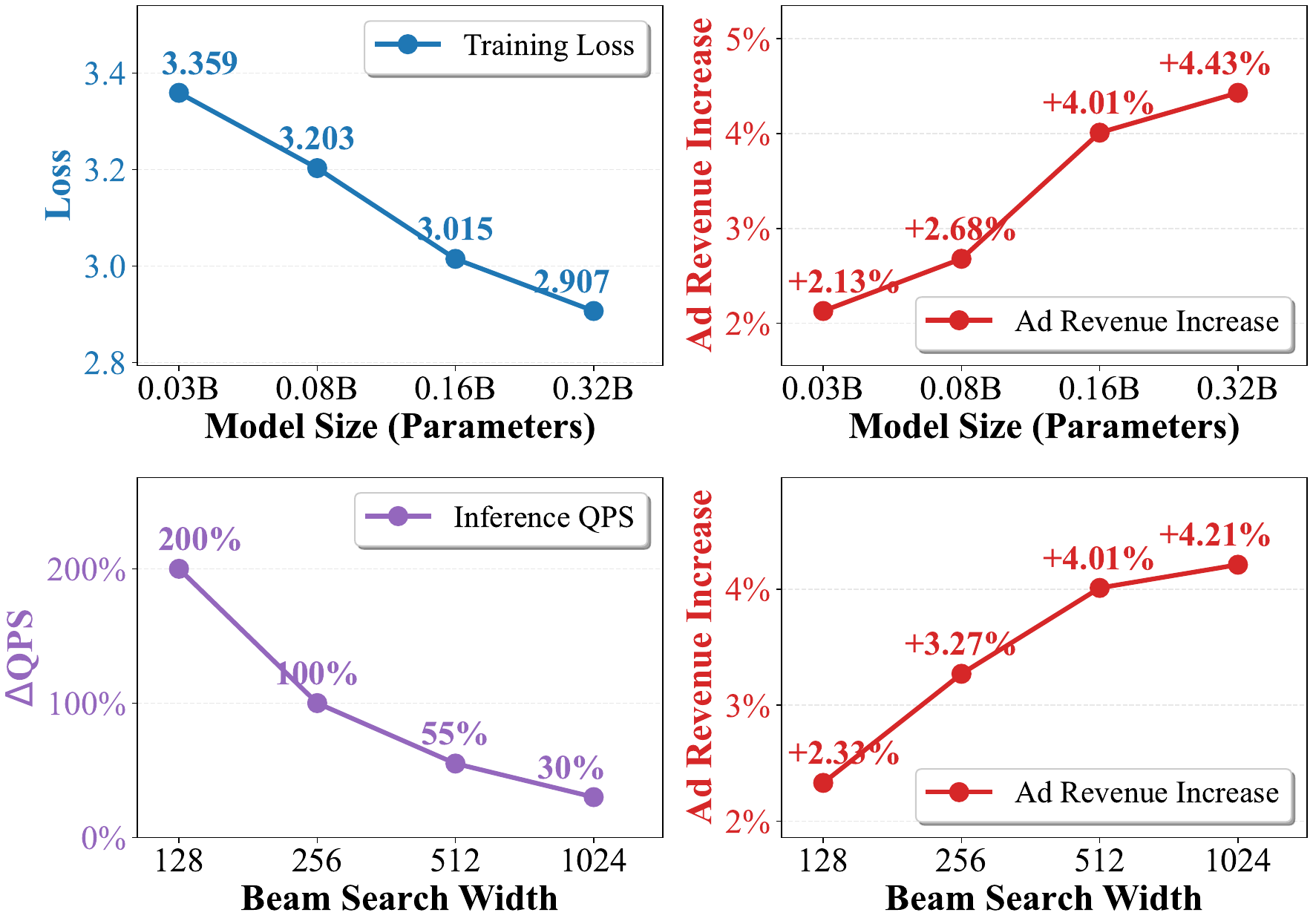}
\vspace{-1em}
    \caption{Scaling laws of model size and beam width.}
    \label{fig:scaling_combined}
\vspace{-1em}
\end{figure}

\subsubsection{Inference Scaling.}
Beyond model size, we further investigate inference-time scaling by increasing the beam search width to enhance candidate exploration during generation. Online results show consistent gains with wider beams at a fixed model size of 0.16B: the revenue lift increases from +2.33\% at beam width 128 to +4.21\% at beam width 1024. These results indicate that stronger inference-time search can further unlock the potential of generative recommenders and translate into meaningful business impact. In production, we select the beam width based on latency and compute budgets to balance online gains against serving cost.

Overall, we observe a clear Scaling Law in both model size and beam width, which provides valuable insights for balancing resource consumption and performance gains.
%

\subsection{Quality of UA-SID}

\paragraph{Embedding Optimizations}
A high-quality underlying embedding should effectively capture both an advertisement's distinctiveness and the relationships among advertisements. Therefore, we constructed an offline test set to evaluate photo-to-photo recall: a retrieval is considered successful if the two photos belong to the same advertised product.
As illustrated in Table\ref{tab:UA-Emb}, QARM\cite{luo2024qarm} achieves the poorest performance because it summarizes only the video content without accounting for advertisement-specific information. Understanding these cues with Qwen3‑VL‑7B~\cite{qwen3vl} yields a substantial recall improvement. Further, after applying instruction tuning (IT) and co‑occurrence learning (CL) to Qwen3‑VL‑7B, we obtain the best recall results, underscoring the importance of end-to-end fine‑tuning of advertisement embedding models. More clustering visualization results are shown in Appendix.

\paragraph{Quantization Optimizations}
Under the same quantization space (i.e., the product of codebook sizes across all levels), we regard a SID scheme as better if it achieves a \emph{lower compression ratio} and a \emph{lower collision rate}, since this indicates higher code utilization and stronger item discriminability while still grouping semantically similar items. We report three metrics, compression ratio (Cpr), collision rate (Col) and codebook utilization (Util). The detailed computation is listed in Appendix.
The ablation results are summarized in Table~\ref{tab:UA-Emb}. We observe that, under a fixed quantization space, allocating larger codebooks to earlier levels and smaller ones to later levels (i.e., a multi-resolution design) improves codebook utilization and reduces collisions. Furthermore, applying a randomized hashing strategy at the final layer (denoted * in table) further lowers the collision rate, and can also facilitate learning collaborative structure among items beyond pure semantic partitioning. 

In summary, optimization of the embeddings and quantization yielded a +0.24\% increase in revenue, as illustrated in Table~\ref{tab:ablation_main}.

\begin{table}[t]
\centering

\caption{Ablation of UA-SID.}
\vspace{-1em}
\label{tab:UA-Emb}
\resizebox{1.0\columnwidth}{!}{%
\begin{tabular}{lccc}
\toprule
    \textbf{Tokenization Settings}     &  \multicolumn{3}{c}{\textbf{Offline Metrics}}   \\ 
\midrule
\textit{Embedding Optimizations} & \textbf{R@1}$\uparrow$ & \textbf{R@5}$\uparrow$ & \textbf{R@10}$\uparrow$  \\[2pt]
 \ \ \ \ QARM~\cite{luo2024qarm}                & 0.541             & 0.812 & 0.893      \\
 \ \ \ \ Qwen3-VL-7B~\cite{qwen3vl}          &0.769       & 0.948              & 0.977       \\
 \ \ \ \ Qwen3-VL-7B + IT + CL (UAE) & \textbf{0.896} & \textbf{0.985} & \textbf{0.994} \\

\midrule
\textit{Quantization Optimizations} & \textbf{Cpr}$\downarrow$ & \textbf{Col}$\downarrow$ & \textbf{Util}$\uparrow$ \\[2pt]

\ \ \ \ RQ-Kmeans ~\cite{luo2024qarm, zhou2025onerec}  & \multirow{2}{*}{3.54}  & \multirow{2}{*}{85.44\%} & \multirow{2}{*}{0.10\textperthousand} \\
\ \ \ \ \ \ \ \ (4096, 4096, 4096) \\
\ \ \ \ RQ-Kmeans + MR & \multirow{2}{*}{1.78} & \multirow{2}{*}{59.72\%} & \multirow{2}{*}{0.20\textperthousand} \\
\ \ \ \ \ \ \ \ (16384, 4096, 1024) \\
\ \ \ \ RQ-Kmeans + MG + MR (UA-SID) & \multirow{2}{*}{\textbf{1.07}} & \multirow{2}{*}{\textbf{18.26\%}}  & \multirow{2}{*}{\textbf{0.34\textperthousand}} \\
\ \ \ \ \ \ \ \ \textbf{(16384, 4096, 1024*)} \\

\bottomrule
\end{tabular}
}
\vspace{-1em}
\end{table}

\section{Conclusions}

In this paper, we presented GR4AD, a production-oriented generative recommender for large-scale, real-time advertising with online learning. GR4AD is co-designed across tokenization, architecture, learning and serving: UA-SID facilitates accurate advertisement representation; LazyAR improves the efficiency of short, multi-candidate generation by relaxing layer-wise dependencies while preserving effectiveness; VSL and the proposed RSPO align continual optimization with business value in non-stationary markets; and dynamic beam serving adapts decoding compute to traffic fluctuations and latency budgets. Large-scale online A/B tests show up to 4.2\% ad revenue improvement over a strong DLRM baseline, with consistent gains from both model scaling and inference-time scaling, while achieving high-throughput real-time serving in a fully deployed system. These results highlight the promise of recommendation-native generative design for advertising and suggest future directions in more robust continual learning, cost-aware inference control, and broader constraint optimization in real-world deployments.

\clearpage
\balance

\section*{Acknowledgments}
We sincerely thank the following individuals (listed in alphabetical order) for their invaluable contributions: Caiyi Xu, Chen Yang, Chen Li, Fuxing Zhang, Haiping Xu, Hongtao Cheng, Jin Ouyang, Jinghui Jia, Jingshan Lv, Kang Sun, Lejian Ren, Qigen Hu, Xiang He, Xin Ku, Xinchen Luo, Yiyu Wang, Yongchuan Wang, Zhan Hu, Zhaojie Liu, Zhongteng Han

\bibliographystyle{ACM-Reference-Format}
\bibliography{reference}

@article{deepseek2024dv3,
  author       = {DeepSeek{-}AI},
  title        = {DeepSeek-V3 Technical Report},
  journal      = {CoRR},
  volume       = {abs/2412.19437},
  year         = {2024}
}

@article{zhang2025policy,
  author       = {Wenhao Zhang and
                  Yuexiang Xie and
                  Yuchang Sun and
                  Yanxi Chen and
                  Guoyin Wang and
                  Yaliang Li and
                  Bolin Ding and
                  Jingren Zhou},
  title        = {On-Policy {RL} Meets Off-Policy Experts: Harmonizing Supervised Fine-Tuning
                  and Reinforcement Learning via Dynamic Weighting},
  journal      = {CoRR},
  volume       = {abs/2508.11408},
  year         = {2025}
}

@article{fu2025srft,
  author       = {Yuqian Fu and
                  Tinghong Chen and
                  Jiajun Chai and
                  Xihuai Wang and
                  Songjun Tu and
                  Guojun Yin and
                  Wei Lin and
                  Qichao Zhang and
                  Yuanheng Zhu and
                  Dongbin Zhao},
  title        = {{SRFT:} {A} Single-Stage Method with Supervised and Reinforcement
                  Fine-Tuning for Reasoning},
  journal      = {CoRR},
  volume       = {abs/2506.19767},
  year         = {2025}
}

@article{zhou2025onerec,
  author       = {Guorui Zhou and
                  Jiaxin Deng and
                  Jinghao Zhang and
                  Kuo Cai and
                  Lejian Ren and
                  Qiang Luo and
                  Qianqian Wang and
                  Qigen Hu and
                  Rui Huang and
                  Shiyao Wang and
                  Weifeng Ding and
                  Wuchao Li and
                  Xinchen Luo and
                  Xingmei Wang and
                  Zexuan Cheng and
                  Zixing Zhang and
                  Bin Zhang and
                  Boxuan Wang and
                  Chaoyi Ma and
                  Chengru Song and
                  Chenhui Wang and
                  Di Wang and
                  Dongxue Meng and
                  Fan Yang and
                  Fangyu Zhang and
                  Feng Jiang and
                  Fuxing Zhang and
                  Gang Wang and
                  Guowang Zhang and
                  Han Li and
                  Hengrui Hu and
                  Hezheng Lin and
                  Hongtao Cheng and
                  Hongyang Cao and
                  Huanjie Wang and
                  Jiaming Huang and
                  Jiapeng Chen and
                  Jiaqiang Liu and
                  Jinghui Jia and
                  Kun Gai and
                  Lantao Hu and
                  Liang Zeng and
                  Liao Yu and
                  Qiang Wang and
                  Qidong Zhou and
                  Shengzhe Wang and
                  Shihui He and
                  Shuang Yang and
                  Shujie Yang and
                  Sui Huang and
                  Tao Wu and
                  Tiantian He and
                  Tingting Gao and
                  Wei Yuan and
                  Xiao Liang and
                  Xiaoxiao Xu and
                  Xugang Liu and
                  Yan Wang and
                  Yi Wang and
                  Yiwu Liu and
                  Yue Song and
                  Yufei Zhang and
                  Yunfan Wu and
                  Yunfeng Zhao and
                  Zhanyu Liu},
  title        = {OneRec Technical Report},
  journal      = {CoRR},
  volume       = {abs/2506.13695},
  year         = {2025}
}

@article{qwen3vl,
  author       = {Shuai Bai and
                  Yuxuan Cai and
                  Ruizhe Chen and
                  Keqin Chen and
                  Xionghui Chen and
                  Zesen Cheng and
                  Lianghao Deng and
                  Wei Ding and
                  Chang Gao and
                  Chunjiang Ge and
                  Wenbin Ge and
                  Zhifang Guo and
                  Qidong Huang and
                  Jie Huang and
                  Fei Huang and
                  Binyuan Hui and
                  Shutong Jiang and
                  Zhaohai Li and
                  Mingsheng Li and
                  Mei Li and
                  Kaixin Li and
                  Zicheng Lin and
                  Junyang Lin and
                  Xuejing Liu and
                  Jiawei Liu and
                  Chenglong Liu and
                  Yang Liu and
                  Dayiheng Liu and
                  Shixuan Liu and
                  Dunjie Lu and
                  Ruilin Luo and
                  Chenxu Lv and
                  Rui Men and
                  Lingchen Meng and
                  Xuancheng Ren and
                  Xingzhang Ren and
                  Sibo Song and
                  Yuchong Sun and
                  Jun Tang and
                  Jianhong Tu and
                  Jianqiang Wan and
                  Peng Wang and
                  Pengfei Wang and
                  Qiuyue Wang and
                  Yuxuan Wang and
                  Tianbao Xie and
                  Yiheng Xu and
                  Haiyang Xu and
                  Jin Xu and
                  Zhibo Yang and
                  Mingkun Yang and
                  Jianxin Yang and
                  An Yang and
                  Bowen Yu and
                  Fei Zhang and
                  Hang Zhang and
                  Xi Zhang and
                  Bo Zheng and
                  Humen Zhong and
                  Jingren Zhou and
                  Fan Zhou and
                  Jing Zhou and
                  Yuanzhi Zhu and
                  Ke Zhu},
  title        = {Qwen3-VL Technical Report},
  journal      = {CoRR},
  volume       = {abs/2511.21631},
  year         = {2025}
}

@inproceedings{JuCNKW0S25,
  author       = {Clark Mingxuan Ju and
                  Liam Collins and
                  Leonardo Neves and
                  Bhuvesh Kumar and
                  Louis Yufeng Wang and
                  Tong Zhao and
                  Neil Shah},
  title        = {Generative Recommendation with Semantic IDs: {A} Practitioner's Handbook},
  booktitle    = {{CIKM}},
  pages        = {6420--6425},
  publisher    = {{ACM}},
  year         = {2025}
}

@article{yang2020large,
  author       = {Xiaoyong Yang and
                  Yadong Zhu and
                  Yi Zhang and
                  Xiaobo Wang and
                  Quan Yuan},
  title        = {Large Scale Product Graph Construction for Recommendation in E-commerce},
  journal      = {CoRR},
  volume       = {abs/2010.05525},
  year         = {2020}
}

@article{zhou2025onerecv2,
  author       = {Guorui Zhou and
                  Hengrui Hu and
                  Hongtao Cheng and
                  Huanjie Wang and
                  Jiaxin Deng and
                  Jinghao Zhang and
                  Kuo Cai and
                  Lejian Ren and
                  Lu Ren and
                  Liao Yu and
                  Pengfei Zheng and
                  Qiang Luo and
                  Qianqian Wang and
                  Qigen Hu and
                  Rui Huang and
                  Ruiming Tang and
                  Shiyao Wang and
                  Shujie Yang and
                  Tao Wu and
                  Wuchao Li and
                  Xinchen Luo and
                  Xingmei Wang and
                  Yi Su and
                  Yunfan Wu and
                  Zexuan Cheng and
                  Zhanyu Liu and
                  Zixing Zhang and
                  Bin Zhang and
                  Boxuan Wang and
                  Chaoyi Ma and
                  Chengru Song and
                  Chenhui Wang and
                  Chenglong Chu and
                  Di Wang and
                  Dongxue Meng and
                  Dunju Zang and
                  Fan Yang and
                  Fangyu Zhang and
                  Feng Jiang and
                  Fuxing Zhang and
                  Gang Wang and
                  Guowang Zhang and
                  Han Li and
                  Honghui Bao and
                  Hongyang Cao and
                  Jiaming Huang and
                  Jiapeng Chen and
                  Jiaqiang Liu and
                  Jinghui Jia and
                  Kun Gai and
                  Lantao Hu and
                  Liang Zeng and
                  Qiang Wang and
                  Qidong Zhou and
                  Rongzhou Zhang and
                  Shengzhe Wang and
                  Shihui He and
                  Shuang Yang and
                  Siyang Mao and
                  Sui Huang and
                  Tiantian He and
                  Tingting Gao and
                  Wei Yuan and
                  Xiao Liang and
                  Xiaoxiao Xu and
                  Xugang Liu and
                  Yan Wang and
                  Yang Zhou and
                  Yi Wang and
                  Yiwu Liu and
                  Yue Song and
                  Yufei Zhang and
                  Yunfeng Zhao and
                  Zhixin Ling and
                  Ziming Li},
  title        = {OneRec-V2 Technical Report},
  journal      = {CoRR},
  volume       = {abs/2508.20900},
  year         = {2025}
}

@inproceedings{zheng2024adapting,
  author       = {Bowen Zheng and
                  Yupeng Hou and
                  Hongyu Lu and
                  Yu Chen and
                  Wayne Xin Zhao and
                  Ming Chen and
                  Ji{-}Rong Wen},
  title        = {Adapting Large Language Models by Integrating Collaborative Semantics
                  for Recommendation},
  booktitle    = {{ICDE}},
  pages        = {1435--1448},
  publisher    = {{IEEE}},
  year         = {2024}
}

@inproceedings{wang2024learnable,
  author       = {Wenjie Wang and
                  Honghui Bao and
                  Xinyu Lin and
                  Jizhi Zhang and
                  Yongqi Li and
                  Fuli Feng and
                  See{-}Kiong Ng and
                  Tat{-}Seng Chua},
  title        = {Learnable Item Tokenization for Generative Recommendation},
  booktitle    = {{CIKM}},
  pages        = {2400--2409},
  publisher    = {{ACM}},
  year         = {2024}
}

@inproceedings{hou2025generating,
  author       = {Yupeng Hou and
                  Jiacheng Li and
                  Ashley Shin and
                  Jinsung Jeon and
                  Abhishek Santhanam and
                  Wei Shao and
                  Kaveh Hassani and
                  Ning Yao and
                  Julian J. McAuley},
  title        = {Generating Long Semantic IDs in Parallel for Recommendation},
  booktitle    = {{KDD} {(2)}},
  pages        = {956--966},
  publisher    = {{ACM}},
  year         = {2025}
}

@article{wang2025nezha,
  author       = {Yejing Wang and
                  Shengyu Zhou and
                  Jinyu Lu and
                  Ziwei Liu and
                  Langming Liu and
                  Maolin Wang and
                  Wenlin Zhang and
                  Feng Li and
                  Wenbo Su and
                  Pengjie Wang and
                  Jian Xu and
                  Xiangyu Zhao},
  title        = {{NEZHA:} {A} Zero-sacrifice and Hyperspeed Decoding Architecture for
                  Generative Recommendations},
  journal      = {CoRR},
  volume       = {abs/2511.18793},
  year         = {2025}
}

@article{xu2025mmq,
  author       = {Yi Xu and
                  Moyu Zhang and
                  Chenxuan Li and
                  Zhihao Liao and
                  Haibo Xing and
                  Hao Deng and
                  Jinxin Hu and
                  Yu Zhang and
                  Xiaoyi Zeng and
                  Jing Zhang},
  title        = {{MMQ:} Multimodal Mixture-of-Quantization Tokenization for Semantic
                  {ID} Generation and User Behavioral Adaptation},
  journal      = {CoRR},
  volume       = {abs/2508.15281},
  year         = {2025}
}

@article{zhou2025openonerec,
  author       = {Guorui Zhou and
                  Shiyao Wang and
                  Shucheng Li and
                  Jiaxing Song and
                  Yujie Lu and
                  Weijieying Ren and
                  Hao Liu and
                  Bo Chen and
                  Mingzhou Zhou and
                  Yu Wang and
                  Yiyan Zhang and
                  Tianshu Wu and
                  Jinze Bai and
                  Xiang Li and
                  Xiangyu Zhao and
                  Xiuqiang He and
                  Zhenhua Dong and
                  Jieming Zhu and
                  Ruiming Tang and
                  Rui Zhang and
                  Xi Xiao and
                  Jun Xu and
                  Zhengyan Zhang and
                  Xuemin Zhao and
                  Shuo Li and
                  Zhiyuan Zhang and
                  Ji-Rong Wen and
                  Weinan Zhang and
                  Tingting Zhang and
                  Jun Wang and
                  Xinxuan Chen and
                  Xin Zhao and
                  Lin Liu and
                  Yifan Liu and
                  Ruihua Song and
                  Jian-Yun Nie and
                  Hongning Wang and
                  Dawei Yin and
                  Hengshu Zhu and
                  Hui Xiong and
                  Depeng Jin and
                  Yong Li and
                  Wenwu Ou and
                  Jian Pei and
                  Bin Cui and
                  Ping Li},
  title        = {OpenOneRec Technical Report},
  journal      = {arXiv preprint arXiv:2512.24762},
  year         = {2025}
}

@inproceedings{hstu2024,
  author       = {Jiaqi Zhai and
                  Lucy Liao and
                  Xing Liu and
                  Yueming Wang and
                  Rui Li and
                  Xuan Cao and
                  Leon Gao and
                  Zhaojie Gong and
                  Fangda Gu and
                  Jiayuan He and
                  Yinghai Lu and
                  Yu Shi},
  title        = {Actions Speak Louder than Words: Trillion-Parameter Sequential Transducers
                  for Generative Recommendations},
  booktitle    = {{ICML}},
  publisher    = {OpenReview.net},
  year         = {2024}
}

@article{liu2026gdpo,
  author    = {Shih-Yang Liu and
               Ta-Chi Yen and
               Cheng-Yu Hsieh and
               Yueh-Ning Chen and
               Chin-Hsuan Lin and
               Yu-Wei Chao and
               Tsu-Jui Fu and
               Shou-De Lin and
               Wan-Ting Hsu and
               Hsiang-Sheng Chiu and
               Pei-Fu Guo and
               Chen-Hsiang Yu and
               Wen-Hsuan Li},
  title     = {GDPO: Group reward-Decoupled Normalization Policy Optimization for Multi-reward RL Optimization},
  journal={arXiv preprint arXiv:2601.05242},
  year         = {2026}
}

@article{chen2025onesearch,
  author       = {Ben Chen and
                  Xian Guo and
                  Siyuan Wang and
                  Zihan Liang and
                  Yue Lv and
                  Yufei Ma and
                  Xinlong Xiao and
                  Bowen Xue and
                  Xuxin Zhang and
                  Ying Yang and
                  Huangyu Dai and
                  Xing Xu and
                  Tong Zhao and
                  Mingcan Peng and
                  Xiaoyang Zheng and
                  Chao Wang and
                  Qihang Zhao and
                  Zhixin Zhai and
                  Yang Zhao and
                  Bochao Liu and
                  Jingshan Lv and
                  Xiao Liang and
                  Yuqing Ding and
                  Jing Chen and
                  Chenyi Lei and
                  Wenwu Ou and
                  Han Li and
                  Kun Gai},
  title        = {OneSearch: {A} Preliminary Exploration of the Unified End-to-End Generative
                  Framework for E-commerce Search},
  journal      = {CoRR},
  volume       = {abs/2509.03236},
  year         = {2025}
}

@article{zhang2025gpr,
  author       = {Jun Zhang and
                  Yi Li and
                  Yue Liu and
                  Changping Wang and
                  Yuan Wang and
                  Yuling Xiong and
                  Xun Liu and
                  Haiyang Wu and
                  Qian Li and
                  Enming Zhang and
                  Jiawei Sun and
                  Xin Xu and
                  Zishuai Zhang and
                  Ruoran Liu and
                  Suyuan Huang and
                  Zhaoxin Zhang and
                  Zhengkai Guo and
                  Shuojin Yang and
                  Meng{-}Hao Guo and
                  Huan Yu and
                  Jie Jiang and
                  Shi{-}Min Hu},
  title        = {{GPR:} Towards a Generative Pre-trained One-Model Paradigm for Large-Scale
                  Advertising Recommendation},
  journal      = {CoRR},
  volume       = {abs/2511.10138},
  year         = {2025}
}

@inproceedings{MTGR2025,
  author       = {Ruidong Han and
                  Bin Yin and
                  Shangyu Chen and
                  He Jiang and
                  Fei Jiang and
                  Xiang Li and
                  Chi Ma and
                  Mincong Huang and
                  Xiaoguang Li and
                  Chunzhen Jing and
                  Yueming Han and
                  MengLei Zhou and
                  Lei Yu and
                  Chuan Liu and
                  Wei Lin},
  title        = {{MTGR:} Industrial-Scale Generative Recommendation Framework in Meituan},
  booktitle    = {{CIKM}},
  pages        = {5731--5738},
  publisher    = {{ACM}},
  year         = {2025}
}

@inproceedings{lambdaloss,
  author       = {Xuanhui Wang and
                  Cheng Li and
                  Nadav Golbandi and
                  Michael Bendersky and
                  Marc Najork},
  title        = {The LambdaLoss Framework for Ranking Metric Optimization},
  booktitle    = {{CIKM}},
  pages        = {1313--1322},
  publisher    = {{ACM}},
  year         = {2018}
}

@inproceedings{scaling2024seqrec,
  author       = {Gaowei Zhang and
                  Yupeng Hou and
                  Hongyu Lu and
                  Yu Chen and
                  Wayne Xin Zhao and
                  Ji{-}Rong Wen},
  title        = {Scaling Law of Large Sequential Recommendation Models},
  booktitle    = {RecSys},
  pages        = {444--453},
  publisher    = {{ACM}},
  year         = {2024}
}

@inproceedings{ouyang2022training,
  author       = {Long Ouyang and
                  Jeffrey Wu and
                  Xu Jiang and
                  Diogo Almeida and
                  Carroll L. Wainwright and
                  Pamela Mishkin and
                  Chong Zhang and
                  Sandhini Agarwal and
                  Katarina Slama and
                  Alex Ray and
                  John Schulman and
                  Jacob Hilton and
                  Fraser Kelton and
                  Luke Miller and
                  Maddie Simens and
                  Amanda Askell and
                  Peter Welinder and
                  Paul F. Christiano and
                  Jan Leike and
                  Ryan Lowe},
  title        = {Training language models to follow instructions with human feedback},
  booktitle    = {NeurIPS},
  year         = {2022}
}

@inproceedings{rafailov2023direct,
  author       = {Rafael Rafailov and
                  Archit Sharma and
                  Eric Mitchell and
                  Christopher D. Manning and
                  Stefano Ermon and
                  Chelsea Finn},
  title        = {Direct Preference Optimization: Your Language Model is Secretly a
                  Reward Model},
  booktitle    = {NeurIPS},
  year         = {2023}
}

@inproceedings{meng2024simpo,
  author       = {Yu Meng and
                  Mengzhou Xia and
                  Danqi Chen},
  title        = {SimPO: Simple Preference Optimization with a Reference-Free Reward},
  booktitle    = {NeurIPS},
  year         = {2024}
}

@article{shao2024deepseekmath,
  author       = {Zhihong Shao and
                  Peiyi Wang and
                  Qihao Zhu and
                  Runxin Xu and
                  Junxiao Song and
                  Mingchuan Zhang and
                  Y. K. Li and
                  Y. Wu and
                  Daya Guo},
  title        = {DeepSeekMath: Pushing the Limits of Mathematical Reasoning in Open
                  Language Models},
  journal      = {CoRR},
  volume       = {abs/2402.03300},
  year         = {2024}
}

@article{gao2025soft,
  author       = {Chang Gao and
                  Chujie Zheng and
                  Xionghui Chen and
                  Kai Dang and
                  Shixuan Liu and
                  Bowen Yu and
                  An Yang and
                  Shuai Bai and
                  Jingren Zhou and
                  Junyang Lin},
  title        = {Soft Adaptive Policy Optimization},
  journal      = {CoRR},
  volume       = {abs/2511.20347},
  year         = {2025}
}

@article{HPT,
  author       = {Xingtai Lv and
                  Yuxin Zuo and
                  Youbang Sun and
                  Hongyi Liu and
                  Yuntian Wei and
                  Zhekai Chen and
                  Lixuan He and
                  Xuekai Zhu and
                  Kaiyan Zhang and
                  Bingning Wang and
                  Ning Ding and
                  Bowen Zhou},
  title        = {Towards a Unified View of Large Language Model Post-Training},
  journal      = {CoRR},
  volume       = {abs/2509.04419},
  year         = {2025}
}

@inproceedings{rajput2023tiger,
  author       = {Shashank Rajput and
                  Nikhil Mehta and
                  Anima Singh and
                  Raghunandan Hulikal Keshavan and
                  Trung Vu and
                  Lukasz Heldt and
                  Lichan Hong and
                  Yi Tay and
                  Vinh Q. Tran and
                  Jonah Samost and
                  Maciej Kula and
                  Ed H. Chi and
                  Mahesh Sathiamoorthy},
  title        = {Recommender Systems with Generative Retrieval},
  booktitle    = {NeurIPS},
  year         = {2023}
}

@inproceedings{dm2020,
  author       = {Jonathan Ho and
                  Ajay Jain and
                  Pieter Abbeel},
  title        = {Denoising Diffusion Probabilistic Models},
  booktitle    = {NeurIPS},
  year         = {2020}
}

@article{gibbs,
  author       = {Gelfand and
                  Alan E},
  title        = {Gibbs sampling},
  journal      = {Journal of the American statistical Association},
  volume       = {95},
  number       = {452},
  pages        = {1300--1304},
  year         = {2000},
  publisher    = {Taylor \& Francis}
}

@article{gpt,
  author       = {Radford Alec and
                  Narasimhan Karthik and
                  Salimans Tim and
                  Sutskever Ilya},
  title        = {Improving language understanding by generative pre-training},
  year         = {2018},
  publisher    = {San Francisco, CA, USA}
}

@article{vae2019,
  author       = {Diederik P. Kingma and
                  Max Welling},
  title        = {An Introduction to Variational Autoencoders},
  journal      = {Found. Trends Mach. Learn.},
  volume       = {12},
  number       = {4},
  pages        = {307--392},
  year         = {2019}
}

@inproceedings{gan,
  author       = {Ian J. Goodfellow and
                  Jean Pouget{-}Abadie and
                  Mehdi Mirza and
                  Bing Xu and
                  David Warde{-}Farley and
                  Sherjil Ozair and
                  Aaron C. Courville and
                  Yoshua Bengio},
  title        = {Generative Adversarial Nets},
  booktitle    = {{NIPS}},
  pages        = {2672--2680},
  year         = {2014}
}

@inproceedings{luo2024qarm,
  author       = {Xinchen Luo and
                  Jiangxia Cao and
                  Tianyu Sun and
                  Jinkai Yu and
                  Rui Huang and
                  Wei Yuan and
                  Hezheng Lin and
                  Yichen Zheng and
                  Shiyao Wang and
                  Qigen Hu and
                  Changqing Qiu and
                  Jiaqi Zhang and
                  Xu Zhang and
                  Zhiheng Yan and
                  Jingming Zhang and
                  Simin Zhang and
                  Mingxing Wen and
                  Zhaojie Liu and
                  Guorui Zhou},
  title        = {{QARM:} Quantitative Alignment Multi-Modal Recommendation at Kuaishou},
  booktitle    = {{CIKM}},
  pages        = {5915--5922},
  publisher    = {{ACM}},
  year         = {2025}
}

@inproceedings{wang2024content,
  author       = {Yidan Wang and
                  Zhaochun Ren and
                  Weiwei Sun and
                  Jiyuan Yang and
                  Zhixiang Liang and
                  Xin Chen and
                  Ruobing Xie and
                  Su Yan and
                  Xu Zhang and
                  Pengjie Ren and
                  Zhumin Chen and
                  Xin Xin},
  title        = {Content-Based Collaborative Generation for Recommender Systems},
  booktitle    = {{CIKM}},
  pages        = {2420--2430},
  publisher    = {{ACM}},
  year         = {2024}
}

@inproceedings{si2024generative,
  author       = {Zihua Si and
                  Zhongxiang Sun and
                  Jiale Chen and
                  Guozhang Chen and
                  Xiaoxue Zang and
                  Kai Zheng and
                  Yang Song and
                  Xiao Zhang and
                  Jun Xu and
                  Kun Gai},
  title        = {Generative Retrieval with Semantic Tree-Structured Identifiers and
                  Contrastive Learning},
  booktitle    = {{SIGIR-AP}},
  pages        = {154--163},
  publisher    = {{ACM}},
  year         = {2024}
}

@article{qu2024tokenrec,
  author       = {Haohao Qu and
                  Wenqi Fan and
                  Zihuai Zhao and
                  Qing Li},
  title        = {TokenRec: Learning to Tokenize {ID} for LLM-Based Generative Recommendations},
  journal      = {{IEEE} Trans. Knowl. Data Eng.},
  volume       = {37},
  number       = {10},
  pages        = {6216--6231},
  year         = {2025}
}

@article{fu2025forge,
  author       = {Kairui Fu and
                  Tao Zhang and
                  Shuwen Xiao and
                  Ziyang Wang and
                  Xinming Zhang and
                  Chenchi Zhang and
                  Yuliang Yan and
                  Junjun Zheng and
                  Yu Li and
                  Zhihong Chen and
                  Jian Wu and
                  Xiangheng Kong and
                  Shengyu Zhang and
                  Kun Kuang and
                  Yu{-}Gang Jiang and
                  Bo Zheng},
  title        = {{FORGE:} Forming Semantic Identifiers for Generative Retrieval in
                  Industrial Datasets},
  journal      = {CoRR},
  volume       = {abs/2509.20904},
  year         = {2025}
}

@inproceedings{zhu2024cost,
  author       = {Jieming Zhu and
                  Mengqun Jin and
                  Qijiong Liu and
                  Zexuan Qiu and
                  Zhenhua Dong and
                  Xiu Li},
  title        = {CoST: Contrastive Quantization based Semantic Tokenization for Generative
                  Recommendation},
  booktitle    = {RecSys},
  pages        = {969--974},
  publisher    = {{ACM}},
  year         = {2024}
}

@article{he2025plum,
  author       = {Ruining He and
                  Lukasz Heldt and
                  Lichan Hong and
                  Raghunandan H. Keshavan and
                  Shifan Mao and
                  Nikhil Mehta and
                  Zhengyang Su and
                  Alicia Tsai and
                  Yueqi Wang and
                  Shao{-}Chuan Wang and
                  Xinyang Yi and
                  Lexi Baugher and
                  Baykal Cakici and
                  Ed H. Chi and
                  Cristos Goodrow and
                  Ningren Han and
                  He Ma and
                  R{\'{o}}mer Rosales and
                  Abby Van Soest and
                  Devansh Tandon and
                  Su{-}Lin Wu and
                  Weilong Yang and
                  Yilin Zheng},
  title        = {{PLUM:} Adapting Pre-trained Language Models for Industrial-scale
                  Generative Recommendations},
  journal      = {CoRR},
  volume       = {abs/2510.07784},
  year         = {2025}
}

@inproceedings{SDPO,
  author       = {Yuxin Chen and
                  Junfei Tan and
                  An Zhang and
                  Zhengyi Yang and
                  Leheng Sheng and
                  Enzhi Zhang and
                  Xiang Wang and
                  Tat{-}Seng Chua},
  title        = {On Softmax Direct Preference Optimization for Recommendation},
  booktitle    = {NeurIPS},
  year         = {2024}
}

@inproceedings{FP8,
  author       = {Maxim Fishman and
                  Brian Chmiel and
                  Ron Banner and
                  Daniel Soudry},
  title        = {Scaling {FP8} training to trillion-token LLMs},
  booktitle    = {{ICLR}},
  publisher    = {OpenReview.net},
  year         = {2025}
}
\clearpage

\appendix

\section{Appendix}
\subsection{Connection Between RSPO and NDCGcost}
\label{app:rspo}


The definition of NDCG provided in the paper is restated as follows: 
\begin{align}
\mathrm{NDCG} = \frac{1}{Z} \sum_{i=1}^{n} \frac{2^{v_i} - 1}{\log_{2}(1 + i)}
= \sum_{i=1}^{n} \frac{G_i}{D_i}
\end{align}

where $Z$ is the ideal DCG, $v_i$ denotes the associated reward (eCPM) for sample $i$, $G_i = \frac{2^{v_i} - 1}{Z}$, $D_i = \log_{2}(1 + i)$. For any sample pair \( (\mathbf{y}_i, \mathbf{y}_j) \) with $eCPM_i > eCPM_j$. In RSPO, let the ranking score based on NDCG be denoted as: 

\begin{equation}
\left\{
\frac{\log p_\theta(\mathbf{y}_i \mid X)}{\log p_{ref}(\mathbf{y}_i \mid X)^{C_{ij}}},
\frac{\log p_\theta(\mathbf{y}_j \mid X)}{\log p_{ref}(\mathbf{y}_j \mid X)^{C_{ij}}}
\right\} = \left\{ g_i, g_j\right\}
\end{equation}

Following the LambdaLoss\cite{lambdaloss} framework, NDCGcost is defined as follows:
\begin{equation}
\begin{aligned}
\mathrm{NDCGcost} 
&= \sum_{i=1}^{n} G_i - \sum_{i=1}^{n} \frac{G_i}{D_i} \\
&= \sum_{i=1}^{n} G_i \sum_{j=1}^{n} \left| \frac{1}{D_{|i-j|}} - \frac{1}{D_{|i-j|+1}} \right| \mathbb{I}_{g_i < g_j} \\
&= \sum_{i=1}^{n} \sum_{\mathbf{y}_j \in \mathcal{E}_i} \left| \frac{1}{D_{|i-j|}} - \frac{1}{D_{|i-j|+1}} \right| |G_i - G_j| 
    \mathbb{I}_{g_i < g_j} + Const \\
&= \sum_{i=1}^{n} \sum_{\mathbf{y}_j \in \mathcal{E}_i} \mathcal{M}_{ij} \mathbb{I}_{g_i < g_j} + Const
\end{aligned}
\end{equation}

Where $\mathcal{M}_{ij} = \left| \frac{1}{D_{|i-j|}} - \frac{1}{D_{|i-j|+1}} \right| |G_i - G_j|$. $\mathbb{I}_{g_i < g_j}$ is the zero-one indicator. $Const$ is a constant, which is only related to the current ranking. $\mathcal{E}_i$ denotes the set of negative samples corresponding to sample $i$, where $v_j<v_i$. 

\begin{theorem}
Ignoring constant terms, $\mathcal{L}_\text{RSPO}$ optimizes an upper bound of $\mathrm{NDCGcost}$.
\begin{equation}
\begin{aligned}
    \mathrm{NDCGcost} = \sum_{i=1}^{n} \sum_{\mathbf{y}_j \in \mathcal{E}_i} \mathbb{I}_{g_i < g_j} \mathcal{M}_{ij} 
    \le 
    \sum_{i=1}^{n} \sum_{\mathbf{y}_j \in \mathcal{E}_i} -\log_{2} \sigma\left( g_i - g_j \right) \mathcal{M}_{ij} \\
    \le 
    \sum_{i=1}^{n}-\log_{2} \sigma \Bigg( -\log \sum_{\mathbf{y}_j \in \mathcal{E}_i} \mathcal{M}_{ij} \exp(g_j - g_i)\Bigg) = \mathcal{L}_\text{RSPO}
\end{aligned}
\end{equation}

\label{thm:rspo}
\end{theorem}

\begin{proof}
For the proof of the first inequality, we directly follow the inequality proposed in Lambdaloss, that is $\mathbb{I}_{g_i < g_j} \leq -\log_{2} \sigma\left( g_i - g_j  \right)$. For the proof of the second inequality, we prove the inequality for each sample $i$ individually. Let $\ell_i^{(1)}$ and $\ell_i^{(2)}$ denote the inner terms for sample $i$, respectively.
Using the identity $-\log_{2} \sigma(x) = \log_{2}(1 + e^{-x})$, we rewrite $\ell_i^{(1)}$ and $\ell_i^{(2)}$:
\begin{equation}
    \ell_i^{(1)} = \sum_{\mathbf{y}_j \in \mathcal{E}_i} \mathcal{M}_{ij} \log_{2}(1 + e^{-(g_i - g_j)}) = \sum_{\mathbf{y}_j \in \mathcal{E}_i} \mathcal{M}_{ij} \log_{2}(1 + e^{g_j - g_i})
\end{equation}

\begin{equation}
    \ell_i^{(2)} = -\log_{2} \sigma(-\log \sum_{\mathbf{y}_j \in \mathcal{E}_i} \mathcal{M}_{ij} e^{g_j - g_i}) = \log_{2} \left( 1 + \sum_{\mathbf{y}_j \in \mathcal{E}_i} \mathcal{M}_{ij} e^{g_j - g_i} \right)
\end{equation}

To prove ~\ref{thm:rspo}, it is sufficient to demonstrate that $\ell_i^{(1)} \le \ell_i^{(2)}$. According to the definition of $\mathcal{M}_{ij}$, we have:
\begin{align}
    \sum_{\mathbf{y}_j \in \mathcal{E}_j} \mathcal{M}_{ij}
    &\le \sum_{\mathbf{y}_j \in \mathcal{E}_j}
        \left| \frac{1}{D_{|i-j|}} - \frac{1}{D_{|i-j|+1}} \right| \notag \\
    &= \sum_{k \in [b_1,b_2]}
        \left| \frac{1}{D_{k}} - \frac{1}{D_{k+1}} \right| = \left| \frac{1}{D_{b_1}} - \frac{1}{D_{b_2}} \right| \notag \\
    &= \left| \frac{1}{\log_{2}(1 + b_{1})}
        - \frac{1}{\log_{2}(1 + b_{2})} \right|  < 1
\end{align}

Let $k$ be the distance between $i$ and $j$, such that $k \in [b_1,b_2]$ with $b_{1},b_{2} \ge 1$. Introduce a dummy variable to normalize the weights. Let us define a complementary weight $\mathcal{M}_{i0}$ such that:
\begin{equation}
    \mathcal{M}_{i0} =  1 - \sum_{\mathbf{y}_j \in \mathcal{E}_i} \mathcal{M}_{ij}.
\end{equation} 

Let $x_j = e^{g_j - g_i}$. Define the function $f(x) = \log_{2}(1+x)$ for $x \ge 0$. Its second derivative is $f''(x) = -\frac{1}{\ln 2 (1 + x)^2} < 0$, which implies that $f(x)$ is strictly \textbf{concave}. We apply \textbf{Jensen's Inequality} (where $\mathbb{E}[f(X)] \le f(\mathbb{E}[X])$) to the set of values $\{x_j\}_{\mathbf{y}_j \in \mathcal{E}_i} \cup \{x_0\}$, where we choose $x_0 = 0$. The inequality becomes:
\begin{align}
    \mathcal{M}_{i0} \log_{2}(1 + x_0) + \sum_{\mathbf{y}_j \in \mathcal{E}_i} \mathcal{M}_{ij} \log_{2}(1 + x_j) \le \notag \\
    \log_{2} \left( 1 + \left( \mathcal{M}_{i0} x_0 + \sum_{\mathbf{y}_j \in \mathcal{E}_i} \mathcal{M}_{ij} x_j \right) \right)
\end{align}

Which is exactly:
\begin{equation}
    \sum_{\mathbf{y}_j \in \mathcal{E}_i} \mathcal{M}_{ij} \log_{2}(1 + x_j) \le \log_{2} \left( 1 + \sum_{\mathbf{y}_j \in \mathcal{E}_i} \mathcal{M}_{ij} x_j \right)
\end{equation}

This completes the proof.

\end{proof}

\subsection{Efficiency Optimizations}
Table~\ref{tab:efficiency} presents the key practical optimization points and their corresponding improvements in QPS.

\begin{table}[htbp]
  \centering
  \caption{Practical Inference Efficiency Optimizations}
  \label{tab:efficiency}
  \begin{tabular}{l r}
    \toprule
    Key Optimization (vs. GR4AD) & $\Delta$ QPS \\
    \midrule
    Beam-shared KV Cache     & +212.5\%  \\
    TopK Pre-Cut             & +184.8\%  \\
    Low-Precision Inference  & +50.3\%   \\
    Reco Result Cache        & +27.8\%   \\
    \bottomrule
  \end{tabular}
\end{table}

\paragraph{KV Cache} 

The KV Cache is a widely recognized technique used in the inference process of LLMs. During next-token prediction, it retains the keys and values of self/cross-attention from previous steps, significantly reducing redundant computation and memory access. This optimization is lossless with respect to model outputs and delivers a $2.5\times$ throughput improvement over a non-cached baseline.
Building on KV caching, our proposed Beam-shared KV Caching further increases queries per second (QPS) by an additional $25\%$, yielding an overall inference speedup of more than $3\times$ compared with the non-cached baseline.
\paragraph{TopK Pre-Cut}:
During beam search inference, at step $i$ we select $b_i$ candidates from $b_{i-1} \times |V_i|$ possibilities. Instead of a direct global selection, we first select $b_i$ candidates from $|V_i|$ in parallel for each of the $b_{i-1}$ beams, then choose the final $b_i$ from the resulting $b_{i-1} \times b_i$ candidates. Since $|V_i| \gg b_i$ in practice, this increases GPU parallelism for top-$k$ and reduces comparisons, yielding significant speedups without changing the results.
\paragraph{Low-Precision Inference}
During inference, we adopt FP8 precision~\cite{FP8} to replace FP32 computation, introducing negligible loss in output quality. A/B testing shows a marginal revenue change of approximately -0.1\%, while inference throughput improves by 50.3\%.
\paragraph{Reco Result Cache}
We implement the result caching mechanism and serve 27.8\% of requests directly from the cache within the defined time window (one minute), thereby improving the overall performance of the inference service.

\paragraph{Dynamic Beam Width}
Our investigation into the effect of beam width across layers reveals that its influence on prediction performance escalates progressively. The model exhibits the highest prediction reliability at the initial layers, with later layers requiring wider beams for refined decision-making (Table~\ref{tab:DBW}). This insight directly motivates the design of our Dynamic Beam Width (DBW) mechanism, which allocates computational resources more efficiently across the model depth.

\begin{table}[h]
\centering
\caption{Comparison of Dynamic Beam Width Strategies}
\vspace{-1em}
\label{tab:DBW}
\begin{tabular}{lccc}
\toprule
\textbf{DBW Strategy} & \textbf{Beam Width} & \textbf{$\Delta$Revenue} & \textbf{$\Delta$QPS} \\ 
\midrule
GR-Base & [512,512,512] & \multicolumn{1}{c}{--} & \multicolumn{1}{c}{--} \\
\midrule
1st-level Reduction & [128,512,512] & $-0.10\%$ & $+27\%$ \\
2nd-level Reduction & [512,128,512] & $-0.23\%$ & $+27\%$ \\
3rd-level Reduction & [512,512,128] & $-0.85\%$ & $+5\%$ \\
\midrule
Progressive Increasing & [128,256,512] & \textbf{$-0.15\%$} & \textbf{+45}\%   \\
\bottomrule
\end{tabular}
\vspace{-1em}
\end{table}

\subsection{Unified Advertisement SID}

\begin{figure*}[t]
        \centering
        \includegraphics[width=0.9\linewidth]{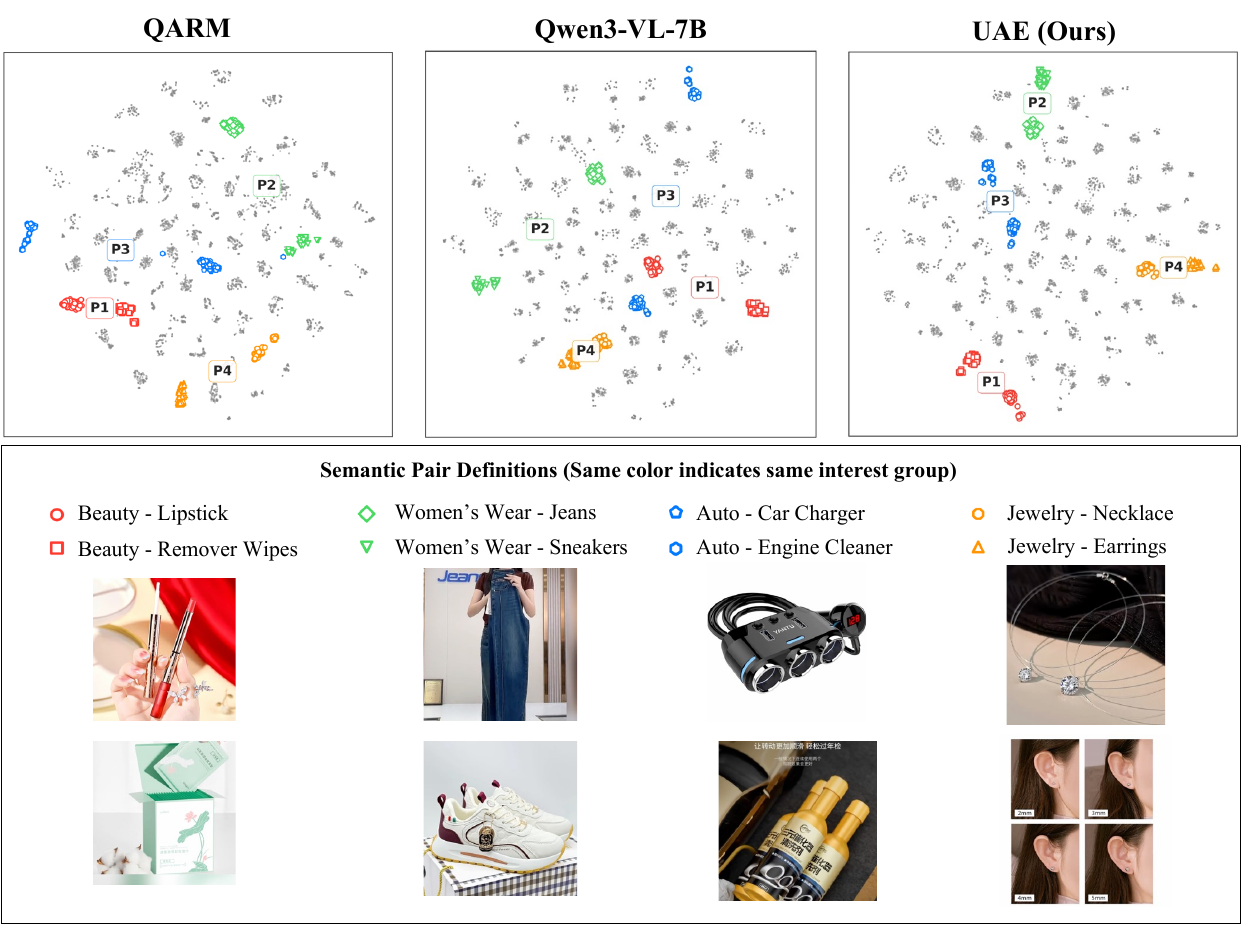}
    \caption{t-SNE visualization of different item embeddings. `P*' represents the clustering center of interest group.}
    \label{fig:visualizations}
\end{figure*}

\subsubsection{SID Evaluation Offline Metrics}

{\small
\begin{align}
\mathrm{Cpr} = \frac{\text{\# Item}}{\text{\# SID}},
\end{align}
\begin{align}
\mathrm{Col} = 1-\frac{\text{\# one-on-one SID}}{\text{\# SID}},
\end{align}
\begin{align}
\mathrm{Util} = \frac{\text{\# Item}}{\text{\# Codebook Space}},
\end{align}
}
where `one-on-one SID' represents the SID associated with only one item.

\subsubsection{Intruction Tuning}

\paragraph{Template 1}
Please perform deep semantic analysis on the input structured text (covering physical products, virtual services, and general entertainment content). Thoroughly fuse discrete fields such as product name, industry classification, detailed category, and brand to accurately determine the vertical domain. Ignore morphological differences and focus on extracting the core category identity and key value attributes (e.g., functional characteristics for physical goods or thematic intent for virtual services), constructing a general semantic representation that precisely defines their essence.

\paragraph{Template 2}
Integrate livestream visual screenshots with multi-dimensional textual information (including streamer profile, regional attributes, livestream title, and user comments/interaction) for deep content understanding. Identify and extract the livestream’s core theme, the streamer’s persona/style, and the current interaction/engagement atmosphere, filtering out nonessential chatter and noise. Construct a semantic representation of the livestream scenario suitable for precise matching and retrieval.

\paragraph{Template 3}
Combine the provided image sequence with product text fields to accurately identify and extract the product's key identity features. Ignore irrelevant marketing phrasing and focus on the product’s essential attributes (industry, category, brand, and core specifications), constructing high-quality product feature representations for precise matching.

\paragraph{Template 4}
Perform deep analysis of virtual products and content-service advertisements (including short dramas, games, apps, etc.). While accounting for differences in product form, focus uniformly on extracting two core feature types: the “core delivered value” and the “narrative marketing strategy.” Accurately identify whether the offering delivers emotional gratification, utilitarian functionality, or cognitive skill-building; and deeply decompose how conversion is driven through scenario construction (e.g., plot conflict, pain-point simulation) and psychological inducements (e.g., suspense hooks, vision framing). Filter out noise and construct a general semantic representation that precisely reflects its benefit attributes and marketing intent.

\paragraph{Template 5}
Analyze advertising video materials intended to drive traffic to livestreams or directly promote products. Integrate visual dynamics and speech interactions to deeply decompose their traffic-conversion logic and persona atmosphere. Identify the video’s core driving elements: the streamer’s persona appeal, product demonstrations and endorsements, the rhetoric or plot devices used to create suspense/expectation (e.g., teasing major offers, emphasizing price advantages, creating urgency), and the livestream’s conveyed core value (premium assortment, cost-effectiveness, expertise, entertainment interaction). Construct semantic feature representations that accurately reflect the marketing and livestream conversion intentions.

\paragraph{Template 6}
Perform a multi-dimensional deep analysis of input physical-product advertisement videos by combining visual frames, scripted voiceover, and OCR-extracted marketing text. While identifying core product attributes (brand, function, appearance), emphasize analysis of content creativity and marketing techniques: determine the presentation format (e.g., narrative dramatization, stress testing, unboxing, factory-sourcing), extract key marketing hooks (e.g., time-limited discounts, pain-point reversals, buy-one-get-one offers), and characterize visual presentation style. Filter out irrelevant noise and construct a high-dimensional semantic representation that encompasses both product features and content-marketing strategy.

\subsubsection{t-SNE Visualization}

Since QARM~\cite{luo2024qarm} lacks ad-specific product information, it can only infer relevance from visual signals, which prevents linking videos with large visual differences that refer to the same product.
Qwen3‑VL‑7B~\cite{qwen3vl} can extract some information but lacks fine‑tuning on real‑world application data and relational signals; after instruction tuning and co‑occurrence learning, the resulting UAE effectively differentiates relationships across interest groups and, within a single interest group, more finely discriminates specific product categories, as illusrated in Figure~\ref{fig:visualizations}.

\end{document}